\documentclass[journal]{IEEEtran}
\usepackage{caption}
\usepackage{subcaption}

\usepackage{amsmath}
\usepackage{amsfonts}
\usepackage{amsthm}
\usepackage{caption}
\usepackage{comment}
\usepackage{amssymb}
\usepackage{tikz}
\usepackage{bbm}
\usepackage{mathtools}
\usepackage{xcolor}

\newtheorem{theorem}{Theorem}
\newtheorem{proposition}{Proposition}
\newtheorem{remark}{Remark}
\newtheorem{definition}{Definition}

\ifCLASSINFOpdf
\else
\fi

\begin{document}
%
\title{Age of Information in a Decentralized Network of Parallel Queues with Routing and Packets Losses}

 
%
%
%

\author{Josu~Doncel and Mohamad~Assaad
\thanks{J. Doncel is with the University of the Basque Country, Spain, 48940 Leioa, Spain (e-mail:josu.doncel@ehu.eus)}
\thanks{M. Assaad is with CentraleSup\'elec, Universit\'e Paris-Saclay, Laboratoire  des  Signaux  et Syst\`emes 91190 Gif-sur-Yvette,  France.}}
%
%

\markboth{Journal of Communications and Networks,~Vol.~XX, No.~XX, XXX~20XX}%
{Doncel \MakeLowercase{\textit{et al.}}: }
%



\markboth{Journal of Communications and Networks,~Vol X, No.~X, Date}%
{Doncel \MakeLowercase{\textit{et al.}}: Status Updates Routing through Decentralized Networks}

\maketitle

\begin{abstract}
The paper deals with Age of Information (AoI) in a network of multiple sources and  parallel queues with buffering capabilities, preemption in service and losses in served packets. The queues do not communicate between each other and the packets are dispatched through the queues according to a predefined probabilistic routing. By making use of the Stochastic Hybrid System (SHS)  method,  we provide a derivation of the average AoI of a system of two  parallel queues (with and without buffer capabilities) and compare the results with those of a single queue. We show that known results of packets delay in Queuing Theory do not hold for the AoI. Unfortunately, the complexity of computing the average AoI using the SHS method  increases highly with the number of queues. We therefore  provide an upper bound of the average AoI in a system of an arbitrary 
number of $M/M/1/(N+1)*$ queues and show its tightness in various regimes. This upper bound allows providing a tight approximation of the average AoI with a very low complexity. We then provide a game framework that allows each source to determine its best probabilistic routing decision. By using Mean Field Games, we provide an analysis of the routing game framework, propose an efficient iterative method to find the routing decision of each source and prove its convergence to the desired equilibrium.
\end{abstract}


\section{Introduction}
Age of Information (AoI) is a relatively new metric that measures the freshness of information in the network. AoI is gaining interest in many areas (e.g. control, communication networks, etc) due to the proliferation of applications in which a monitor is interested in having timely updates about a process of interest. As a typical example, AoI can capture the timeliness of information in a sensor network where the status of a sensor is frequently monitored.   Since its introduction in the seminal papers \cite{kaul2011minimizing,kaul2012real}, AoI has attracted the attention of many researchers in different fields.

A main part of  the  AoI literature focuses on the computation of the average AoI and its minimization, where
the channel/network in which the updates are sent to the monitor is modeled as a queueing system. The computation of AoI in various queueing models have therefore been investigated. For instance,
the authors in \cite{kaul2012real} considered  an M/M/1 queue, an M/D/1 queue or a D/M/1 queue model, and
the authors in \cite{kam2014effect,kam2015effect} studied an M/M/2 queue model. Other queueing models can also be found in \cite{KaulYates20,Moltafet20,kosta2017age,maatouk2018relay,sun2018multipleflows}. While in the aforementioned work, the status updates of the system are assumed to be sent under a predefined transmission policy, the problem of the design of the update policy  has been considered in some papers as well, e.g. \cite{talak2018can}, \cite{sun2017update},\cite{maatouk2019age}. Furthermore, the problem of scheduling  and random access design  with the aim of minimizing the average age of the network has been considered recently in several papers \cite{hsu2017scheduling,bedewy2018scheduling,kadota2018scheduling,maatouk2019noma,maatouk2019csma,maatouk2020csma}. Besides, since  single server queue models are not representative of networks in which packets can be sent through multiple paths, the average AoI has also been studied in networks with parallel servers \cite{bedewy2016age, SunElifKompella18,Y18}, or in more 
complex networks such as  multihop systems \cite{bedewy2019age,bedewy2017age}. For instance, the scheduling of a single packet flow in multi-hop queueing networks was studied in \cite{bedewy2019age,bedewy2017age}. It was shown that Preemptive Last Generated First Served (P-LGFS) policy is age-optimal  if service times are i.i.d. exponentially distributed. In \cite{Farazi19-1}, a multihop scenario in which each node is both a source and a monitor is considered. Fundamental age limits and  near optimal scheduling policies are provided in this work. Recently, the analysis of AoI in a multihop multicast context has been studied in \cite{Buyukates19}.  For more  detailed and comprehensive review of recent work on AoI, one can refer to \cite{sun2020review}.  


In this paper, the network model is different from the aforementioned previous work since we consider  multiple sources that can send their status updates through a system of different parallel  queues. The queues are assumed to be decentralized in the sense that they  cannot communicate between each other. The incoming traffic from each source is dispatched through the parallel queues according to a predefined probabilistic routing. Note that  we  also develop a framework to optimize this probabilistic routing decision as we will see later on in this paper.  Furthermore, we consider a realistic assumption that the transmissions  through the parallel queues are not perfect and that packets can be lost. This assumption is quite realistic in many scenarios, e.g. when the transmission arises over wireless links (which induces errors and hence packet losses) or even in wired networks when the service provider breaks down.   We aim to analyze the average
AoI of this system composed of parallel queues and, for this purpose, we use the stochastic Hybrid Systems (SHS) method, which is introduced in
\cite{YK19} (we explain it in detail in Section~\ref{sec:aoi-shs}). 
A related work to ours is \cite{Y18}, where the authors use the SHS method to compute the average AoI for a system 
formed by multiple sources and an arbitrary number of homogeneous M/M/1/1 queues (i.e. with no buffer)  as well as 
two heterogeneous M/M/1/1 queues, where in both cases preemption in service is allowed. In our work,
we compute the average AoI using the SHS method, including a system formed by two 
heterogeneous M/M/1/2* queues with preemption in service and packet losses. Due to the buffering capability at different queues, the analysis becomes more challenging and complex as compared to \cite{Y18}.  In addition, we assume that queues are decentralized in the sense that they do not communicate between each other. This makes our model different than	\cite{Y18}, where it is assumed that all the queues know where is the freshest update.
Besides, we provide an upper bound of the average AoI in a system of an arbitrary 
number of $M/M/1/(N+1)$ queues. This allows obtaining an approximation of the AoI with a low complexity. Finally, we provide in this paper a game framework to optimize the probabilistic routing decision for each source, which is to the best of our knowledge has not been considered before in the AoI literature.

\begin{table*}[t!]
\centering
  \begin{tabular}{| c | c | c |}
\hline
      & Without losses & With losses\\\hline
     Single Source & (SERVER-ROUTING) and (SERVER-DOUBLE) & (SERVER-ROUTING) and (SERVER-DOUBLE).\\
      &  have equal age for $\lambda$ small and large. See Figure~\ref{fig:comparison_mm11_theta0_alpha0}. & have equal age always. See Figure~\ref{fig:comparison_mm11_theta10_alpha0}. \\\hline
  Multiple Sources & (SERVER-ROUTING) and (SERVER-DOUBLE)& (SERVER-ROUTING) and (SERVER-DOUBLE). \\
&  have almost equal age. See Figure~\ref{fig:comparison_mm11_theta0_alpha10}.   & have equal age. See Figure~\ref{fig:comparison_mm11_theta10_alpha10}.\\\hline
  \end{tabular}
  \caption{Summary of Average AoI comparison of the systems without buffer (see Section~\ref{sec:real-aoi:sub:nobuffer})}
\label{tab:summary-server}.
\end{table*}

\begin{table*}[t!]
\centering
  \begin{tabular}{| c | c | c |}
\hline
      & Without losses & With losses\\\hline
     Single Source & QUEUE-ROUTING and QUEUE-DOUBLE & QUEUE-ROUTING and QUEUE-DOUBLE\\
      & have equal age for $\lambda$ small and large. See Figure~\ref{fig:comparison_mm12_theta0_alpha0}. & have equal age. See Figure~\ref{fig:comparison_mm12_theta10_alpha0}. \\\hline
  Multiple Sources & QUEUE-ROUTING and QUEUE-HALF & QUEUE-ROUTING and QUEUE-DOUBLE \\
& have equal age for $\lambda$ small.  & have equal age always. See Figure~\ref{fig:comparison_mm12_theta10_alpha10}.\\
 &  QUEUE-ROUTING and QUEUE-DOUBLE  &  \\  &have equal age for $\lambda$ large. See Figure~\ref{fig:comparison_mm12_theta0_alpha10}. & \\\hline
  \end{tabular}
  \caption{Summary of Average AoI comparison of the systems with buffer (see Section~\ref{sec:real-aoi:sub:buffer}).}
\label{tab:summary-queue}
\end{table*}

The main contributions of this work are twofold. First, we consider a system with multiple sources where the packets in service can be 
lost and preemption is allowed. The packets are sent to the parallel queues according to a predefined probabilistic routing. 
We compute the average AoI of a system with two parallel queues and we compare its average AoI
with that of a single queue. On one hand, in Table~\ref{tab:summary-server}, we present the obtained results where we compare the following systems: (i) two parallel M/M/1/1 queues, each of them with arrival rate $\lambda/2$, service rate $\mu$ and loss rate $\theta/2$, (denoted by SERVER-ROUTING), (ii) one M/M/1/1 queue  with arrival rate $\lambda/2$, loss rate $\theta/2$ and service rate $\mu$ (denoted by SERVER-HALF) and (iii) one M/M/1/1 queue with arrival rate $\lambda$, loss rate $\theta$ and service rate $2\mu$ (denoted by SERVER-DOUBLE). 
On the other hand, in Table~\ref{tab:summary-queue}, we show  results  in which we compare
the following systems: (i) two parallel M/M/1/2* queues, each of them with arrival rate $\lambda/2$,
service rate $\mu$ and loss rate $\theta$ (denoted by QUEUE-ROUTING), (ii) one
M/M/1/3* queue with arrival rate $\lambda/2$, service rate $\mu$ and loss rate $\theta/2$ 
(denoted by QUEUE-HALF) and (iii) 
one M/M/1/3* queue with arrival rate $\lambda$, service rate $2\mu$ and loss rate $\theta$ (denoted by QUEUE-DOUBLE). The description of M/M/1/3* and M/M/1/2* queues will be provided in Section III. 
The main conclusion of this part of the work is that the known results of packet delay in Queuing Theory do not hold  for the average AoI. 
For instance, we know that the delay of the systems SERVER-ROUTING and SERVER-HALF is the same, whereas according to our results, the average AoI of SERVER-ROUTING is smaller. 
This property also holds when we compare the systems QUEUE-ROUTING and QUEUE-HALF. 
Besides, we also conclude that the average AoI of SERVER-ROUTING is very close to the average AoI of SERVER-DOUBLE and also that the average AoI of QUEUE-ROUTING is very close to the average AoI of QUEUE-DOUBLE.\\ Besides, since the complexity of computing the exact average AoI  with SHS method increases hugely with the number of parallel queues, the second contribution of this work consists of
providing an upper-bound on the average AoI of a system composed of multiple sources  with an arbitrary number of parallel M/M/1/(N+1)* queues.
 We also study numerically the accuracy of the upper bound and we conclude that when the arrival rate is large or 
when there are multiple sources, the upper bound is very tight. The interest of this upper bound lies in the fact that it allows obtaining the average AoI with a low complexity. The last contribution of this work consists in using the derived upper bound of the average AoI  in order to optimize the probabilistic routing decision. For instance, we formulate a distributed framework where each source optimizes its own routing decision using Game Theory. By using Mean Field Games, we then provide a modification/simplification of the framework and derived a simple iterative algorithm allowing each source to find separately its own routing decision. We also provide a theoretical proof of the convergence of the iterative algorithm to the desired fixed point (Equilibrium) of the game.   

The rest of the paper is organized as follows. In Section~\ref{sec:aoi-shs}, we formulate the problem of calculating 
the average AoI and we present how the SHS can be used. In Section~\ref{sec:real-aoi} we focus on the 
average AoI derivation in the different systems under consideration. We present the upper bound of the 
AoI in Section~\ref{sec:bound-aoi} and, finally, we provide the main conclusion of our work in Section~\ref{sec:conclusion}.

\section{AoI and SHS}
\label{sec:aoi-shs}

We consider a transmitter sending status updates to a monitor. Packet $i$ is generated at time $s_i$ and is received by the monitor at 
time $s_i^\prime$. Hence, we define by $N(t)$ the index of the last received update at time $t$, i.e., 
$N(t)= \max\{i | s_i^\prime\leq t\}$, and the time stamp of the last received update at time $t$ as $U(t)=s_{N(t)}$. 
The AoI, or simply the age, is defined as 
$$
\Delta(t)=t-U(t).
$$ 

We are interested in calculating the average of the stochastic process $\Delta(t)$, that is, the average age, 
which is defined as 
$$
\Delta=\lim_{\tau\to\infty}\frac{1}{\tau}\int_0^\tau\Delta(t)dt.
$$

The computation of the average age  in a general setting is known to be a challenging task since the random variables of the 
interarrival times and of the system times are dependent. To overcome this difficulty, the authors in \cite{YK19} introduce 
the SHS. For completeness, we describe hereinafter this method and, for further details one can refer to \cite{hespanha2006modelling}.

In SHS, the system is modeled as a hybrid state $(q(t),{\mathbf x}(t))$, where $q(t)$ is a state of a continuous time Markov 
Chain and ${\mathbf x}(t)$ is a vector whose component belong to $\mathbb R^+_0$ and captures the evolution of the age in the system. 

A link $l$ of the Markov Chain represents a transition from two states $q$ and $q'$ with rate $\lambda^l$. \textcolor{black}{The interest of SHS is that each transition $l$ implies a reset mapping in the continuous process ${\mathbf x}$. In other words, in each transition $l$, the vector 
${\mathbf x}$ is transformed to ${\mathbf x}^\prime$ using a linear mapping where transformation matrix is given by
${\mathbf A}_l$, that is, we have the following SHS transition for every $l$: $\mathbf x^\prime ={\mathbf x}{\mathbf A}_l$. Throughout this paper, we denote by $x_i^\prime$ the i-th element of the vector $\mathbf x^\prime$.}

Furthermore, each state of the Markov Chain represents the elements of the system whose age increases at unit rate. In other words, 
for each state $q$, we define ${\mathbf b}_q$ as the vector whose elements are zero or one. Besides, the evolution of the
vector ${\mathbf x}(t)$ for state $q$ is given by $\dot{\mathbf x}(t)={\mathbf x}{\mathbf b}_q$. 

We assume the Markov Chain is ergodic and we denote by $\pi_q$ the stationary distribution of state $q$. Let ${\mathcal L}_q$ the set
of links that get out of state $q$ and ${\mathcal L}^\prime_q$ the set of links that get into state $q$. The following theorem allows us 
to characterize the average AoI:

\begin{theorem}[{\cite[Thm 4]{YK19}}]
Let $v_q(i)$ denote the $i$-th element of the vector ${\mathbf v}_q$.
For each state $q$, if ${\mathbf v}_q$ is a non-negative solution of the following system of equations
\begin{equation}
{\mathbf v}_q\sum_{l\in{\mathcal L}_q}\lambda^l={\mathbf b}_q\pi_q+\sum_{l \in{\mathcal L}^\prime_q}\lambda^l{\mathbf v}_{q_l}{\mathbf A}_l,
\label{eq:shs-thm}
\end{equation}
then the average AoI is $\Delta=\sum_{q}v_q(0)$.
\end{theorem}

In the following section, we use the above result to characterize the average AoI of several systems. In 
Section~\ref{sec:bound-aoi}, we show that the method under consideration can be also used to obtain an upper-bound on 
the average AoI of very complex systems.

\section{Average AoI of Routing Systems Versus of a Single Queue}
\label{sec:real-aoi}

In this section, we aim to study the average AoI for different configurations using the SHS
method. We first focus on a system formed by queues without buffer and then consider several cases of queues with buffer.   Furthermore, we consider in this section two main scenarios: i) system with single queue, and ii) system with multiple parallel queues. In the latter,  we consider that multiple sources are dispatching their packets through  the different parallel queues according to a predefined probabilistic routing. This kind of routing policies is used in practice and is widely considered in routing literature since it can be implemented without  knowing the instantaneous  states of the network or of the servers (this assumption is realistic as in real-life networks  such information cannot be known at the sources). In more detail, the routing policy can be explained as follows.  Each source $i$ dispatches its packets according to the following policy: each job/packet of the source  is routed to queue $j$ with probability $p_{ij}$. We can see then that the arrival rate from source $i$ to queue $j$ is $\lambda_i p_{ij}$. In addition, it is obvious to see that $\sum_{j=1}^K p_{ij}=1$. Besides, we consider also that the transmission through the queues is not reliable and that the loss rate of each queue $j$ is denoted by $\theta_j$.  We also assume that the queues are decentralized in the sense that they do not communicate between each other.
 
\subsection{Queues Without Buffer}
\label{sec:real-aoi:sub:nobuffer}

In this section, we study the age when the queues do not have a buffer. We first focus on a system formed by an M/M/1/1 queue 
and then in a system with two parallel M/M/1/1 queues.

\subsubsection{The M/M/1/1 queue}
\label{subsubsec:ageoneserver:subsec:nobuffer:sec:real-aoi}
We consider a system formed by one M/M/1/1 queue that receives traffic from different $n$ sources
when preemption of the packets that are in service is permitted. \textcolor{black}{We therefore consider   a Poisson arrival process from each source and hence the resulting arrival from all sources is a Poisson process and two update arrivals cannot occur simultaneously.} 
We are interested in calculating the average AoI of any source.
Without loss of generality, we focus on source $1$. Thus, we consider that updates
	arrive to the system according to a Poisson process, where, without loss of generality, 
	the rate of updates of source $1$ are
	denoted by $\lambda_1$ and of the rest of the sources $\sum_{k>1}\lambda_{k}$. The total arrival rate
	in the system is denoted by $\lambda$, i.e., $\lambda=\sum_{k=1}^n\lambda_k$.
We assume that the service time  is exponentially distributed with rate $\mu$. We
also assume that an update that is in service is lost with an exponential time of rate $\theta$. 

\textcolor{black}{
	We use the SHS method to compute the age of this system. The continuous state of the SHS is ${\mathbf x}(t)=[x_0(t)\ x_1(t)]$, where the correspondence between $x_i$(t) and the elements of the system is 
	as follows: $x_0$ is the age at the monitor and $x_1$ the generation time of the packet in service. One can notice that $x_0(t)= x_1(t)$ if the  packet/update in 
	service is delivered. }The discrete state of the SHS is a two-state Markov Chain, 
where $0$ represents that the system is empty and $1$ that an update is being executed. \textcolor{black}{As explained in the previous section, in each transition in the Markov chain,  the continuous state of the SHS 
${\mathbf x}$ changes to ${\mathbf x}^\prime$.}
The Markov Chain is represented in  Figure~\ref{fig:oneserver} and the SHS transitions 
are given in Table~\ref{tab:oneserver} \textcolor{black}{in which we state explicitly the  transitions from 
${\mathbf x}$  to ${\mathbf x}^\prime$.}

	\begin{table}[t!]
		\centering
			\begin{tabular}{l l l l l}
				$l$ & $q_l\to q_{l^\prime}$ & $\lambda^{l}$ & $x^\prime=x A_l$ & $\bar v_{q_l}A_l$\\\hline
				0 & $0\to 1$ & $\lambda_1$ & $[x_0\ 0]$ & $[v_{0}(0)\ 0]$\\
				1 & $0\to 1$ & $\sum_{k>1}\lambda_k$ & $[x_0\ x_0]$ & $[v_{0}(0)\ v_{0}(0)]$\\
				2 & $1\to 0$ & $\mu$ & $[x_1\ 0]$ & $[v_{1}(1)\ 0]$\\
				3 & $1\to 0$ & $\theta$ & $[x_0\ 0]$ & $[v_{1}(0)\ 0]$\\
				4 & $1\to 1$ & $\lambda_1$ & $[x_0\ 0]$ & $[v_{1}(0)\ 0]$\\
				5 & $1\to 1$ & $\sum_{k>1}\lambda_k$ & $[x_0\ x_0]$ & $[v_{1}(0)\ v_{1}(0)]$\\
		\end{tabular}
		\caption{Table of transitions of Figure~\ref{fig:oneserver}}
		\label{tab:oneserver}
	\end{table}

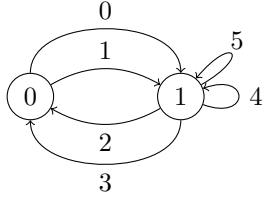
\begin{figure}[t!]
	\centering
	\begin{tikzpicture}[]
	\node[style={circle,draw}] at (0,0) (1) {$0$};
	\node[style={circle,draw}] at (2,0) (2) {$1$};
	\node[] at (3,0) (X) {$4$};
	\node[] at (2.75,0.75) (X) {$5$};
	\draw[->] (2) to [out=340,in=20,looseness=8] (2) ;
	\draw[->] (2) to [out=30,in=50,looseness=18] (2) ;
	\draw[->] (1) edge [out=90,in=90] node[above] {$0$} (2);
	\draw[->] (1) edge [bend left] node[above] {$1$} (2);
	\draw[<-] (1) edge [bend right] node[below] {$2$} (2);
	\draw[<-] (1) edge [out=270, in=270] node[below] {$3$} (2);
	\end{tikzpicture}  
	\caption{The SHS Markov Chain for the one M/M/1/1 queue system with multiple sources and losses
		of packets in service.}
	\label{fig:oneserver}
\end{figure}

We now explain each transition $l$:
\begin{itemize}
	\item[]
	\begin{itemize}
		\item[$l=0$] There is an arrival of source $1$ and the queue is idle. Therefore,
		the age of the monitor does not change, i.e., $x_0^\prime=x_0$ and the age of the packet in service
		is zero since there is a fresh update arrived.
		\item[$l=1$] There is an arrival of one of the others sources when the queue is idle. 
		\textcolor{black}{Since we are interested in the age of source 1, an arrival from the rest of the sources does not bring a fresh update of the status of source $1$ and hence it changes
		the value of $x_1$ to the age of the monitor, that is,
		the age of the update in service satisfies  $x_1^\prime=x_0$, where $x_0$ is the age of the monitor. 
                 Therefore, when this update ends its service, the age of the monitor remains unchanged.}
		\item[$l=2$] The update under execution ends its service and the age of the monitor is
		updated by that of this update, i.e., $x_0^\prime=x_1$.
		\item[$l=3$] The update that is in service is lost and, therefore, the age of the monitor does not
		change.
		\item[$l=4$] There is an arrival of source $1$ when the queue is in service. For this case,
		the update in service is replaced by the fresh one and, therefore,
		the age of the monitor does not change, i.e., $x_0^\prime=x_0$, but the age of the update in service
		changes to zero.
		\item[$l=5$] There is an arrival of another source when the queue is in service. For this case,
		the update in service is replaced by the fresh one and the age of the update in service
		changes to that of the monitor, i.e., $x_1^\prime=x_0$.
	\end{itemize}
\end{itemize}

The stationary distribution
of the Markov Chain of Figure~\ref{fig:oneserver} is 
$$
\pi_0=\frac{\mu+\theta}{\lambda+\mu+\theta},\ \pi_1=\frac{\lambda}{\lambda+\mu+\theta}.
$$

Besides, for the state $q=0$, we have that $\mathbf b_0=[1,0]$ since the age of the monitor 
is the only one that grows at unit rate and the age of the update in service is irrelevant, whereas for the state 
$q=1$ we have that $\mathbf b_1=[1,1]$ and the age of the monitor and of the update in service grow at a unit rate. 
On the other hand, we have that 
$\mathbf v_0=[v_0(0)\ v_0(1)]$ and $\mathbf v_1=[v_1(0)\ v_1(1)]$. From Theorem 4 of 
\cite{YK19}, we know that the age of this system is $v_0(0)+v_1(0)$, where
	\begin{align*}
	\lambda \mathbf{v}_0=&[\pi_0\ 0]+\mu[v_1(1)\ 0]+\theta[v_1(0)\ 0],\\
	(\lambda+\mu+\theta) \mathbf{v}_1=&[\pi_1\ \pi_1]+\lambda_1[v_1(0)\ 0]+\sum_{k>1}\lambda_k[v_1(0)\ v_1(0)]\\&+\lambda_1[v_0(0)\ 0]+\sum_{k>1}\lambda_k[v_0(0)\ v_0(0)],
	\end{align*}
	The above expressions can be written as the following system of equations:
	\begin{subequations}
		\begin{align}
		\lambda v_0(0)&=\pi_0+\mu v_1(1)+\theta v_1(0),\\
		(\mu+\theta) v_1(0)&=\pi_1+\lambda v_0(0),\\
		(\lambda+\mu+\theta) v_1(1)&=\pi_1+\sum_{k>1}\lambda_k v_1(0)+\sum_{k>1}\lambda_k v_0(0).
		\end{align}
		\label{eq:oneserver-shs}
	\end{subequations}

The solution to the above system of equations is
$$
v_0(0)=\frac{1}{\lambda_1}+\frac{\theta}{\lambda_1\mu}-\frac{\pi_1}{\lambda+\mu+\theta}
$$
$$
v_1(0)=\frac{\lambda}{\lambda_1\mu}+\frac{\pi_1}{\lambda+\mu+\theta}
$$
$$
v_1(1)=\frac{\sum_{k>1}\lambda_k}{\lambda_1 \mu} + \frac{\pi_1}{\lambda + \mu + \theta}.
$$
Therefore, since, from \eqref{eq:shs-thm}, the average AoI for this case is $v_0(0)+v_1(0)$, the next result follows:
\begin{proposition}
The average AoI of source $1$ in the aforementioned system is given by 
$$
\frac{1}{\lambda_1}+\frac{\theta}{\lambda_1\mu}+\frac{\lambda}{\lambda_1\mu}.
$$.
\end{proposition}



\begin{remark}
		We remark that, when $\theta=0$, \eqref{eq:oneserver-shs} coincides with the result of Theorem 2a) in 
		\cite{YK19}. In their model, they consider a Markov Chain with a single state, but when
		there are updates that are lost their model cannot be considered. Therefore, we believe that the
		model presented above is the simplest one to study the average AoI 
		using the SHS method when there are update losses.
\end{remark}

\subsubsection{Two parallel M/M/1/1 queues}
	We now consider a system formed by two parallel M/M/1/1 queues receiving traffic from different $n$ sources and where preemption of packets in service is permitted. We aim to calculate the average age of
	information of source $1$. As in the previous case, the arrivals are Poisson and the rate of
	source $1$ is denoted by $\lambda_1$ and that of the rest of the sources is denoted by $\sum_{k>1}\lambda_k$. Hence, $\lambda=\sum_{k=1}^n\lambda_k.$ The packets are dispatched
	according to a predefined probabilistic routing, where $p_{1j}$ (resp. $p_{kj}$) is the probability that a job of source $1$ (resp. 
	of another source $k\neq 1$) is routed to queue $j$. We denote by $\lambda_1 p_{11}$
	and by $\lambda_1 p_{12}$ the arrival rates from source $1$ to queue 1 and to queue 2, respectively.
	Likewise, $\sum_{k>1}\lambda_k p_{k1}$ and $\sum_{k>1}\lambda_k  p_{k2}$ denotes the arrival rates from the rest of the 
	sources to queue 1 and to queue 2, respectively. The  service rate  and the loss rate 
	in queue $j$ are respectively $\mu_j$ and $\theta_j$, where $j=1,2$. We assume that the queues are decentralized
	in the sense that they do not communicate between each other.

\begin{remark}
	The latter assumption makes the model under study here different than
	\cite{Y18}, where it is assumed that the servers know where is the freshest update.
\end{remark}

	We compute the average AoI using the SHS method. First, we define
	the continuous state as $\mathbf x(t)=[x_0(t)\ x_1(t)\ x_2(t)]$, where the correspondence between
	$x_i$(t) and each element in the system is as follows: $x_0$ is the age of the monitor
	and $x_1$ (resp. $x_2$) is the age if an update of queue 1 (resp. of queue 2) is delivered.
The discrete state is a Markov Chain with four states, represented in Figure~\ref{fig:server-routing} and where state $k_1k_2$ represents that in queue 1 there are $k_1$ updates, with $k_1\in\{0,1\}$, 
and in queue 2 there are $k_2$ updates, with $k_2\in\{0,1\}$. We also represent the SHS transitions in 
Table~\ref{tab:server-routing}.

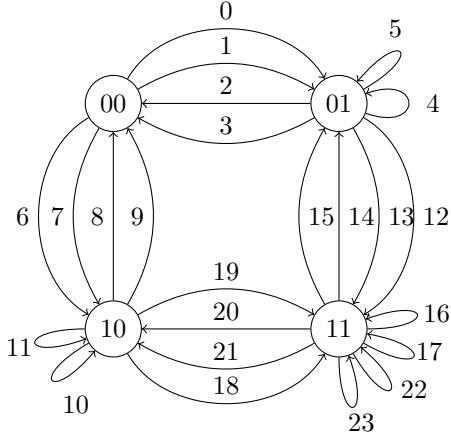
\begin{figure}[t!]
	\centering
	\begin{tikzpicture}[]
	\node[style={circle,draw}] at (0,0) (1) {$00$};
	\node[style={circle,draw}] at (3,0) (2) {$01$};
	\node[style={circle,draw}] at (0,-3) (3) {$10$};
	\node[style={circle,draw}] at (3,-3) (4) {$11$};
	\node[] at (4.25,0) (X1) {$4$};
	\node[] at (3.75,1) (X2) {$5$};
	\node[] at (-0.5,-4) (X3) {$10$};
	\node[] at (-1.25,-3.25) (X4) {$11$};
	\node[] at (4.3,-2.8) (X5) {$16$};
	\node[] at (4.2,-3.3) (X5) {$17$};
	\node[] at (4,-3.8) (X5) {$22$};
	\node[] at (3.3,-4.25) (X5) {$23$};
	\draw[->] (2) to [out=340,in=20,looseness=8] (2) ;
	\draw[->] (2) to [out=30,in=50,looseness=18] (2) ;
	\draw[->] (1) edge [out=60,in=120] node[above] {$0$} (2);
	\draw[->] (1) edge [bend left] node[above] {$1$} (2);
	\draw[<-] (1) edge [<-] node[above] {$2$} (2);
	\draw[<-] (1) edge [bend right] node[above] {$3$} (2);
	\draw[->] (1) edge [out=210,in=150] node[left] {$6$} (3);
	\draw[->] (1) edge [bend right] node[left] {$7$} (3);
	\draw[<-] (1) edge [<-] node[left] {$8$} (3);
	\draw[<-] (1) edge [bend left] node[left] {$9$} (3);
	\draw[->] (3) to [out=180,in=200,looseness=18] (3) ;
	\draw[->] (3) to [out=210,in=230,looseness=18] (3) ;
	\draw[->] (2) edge [out=330,in=30] node[right] {$12$} (4);
	\draw[->] (2) edge [bend left] node[right] {$13$} (4);
	\draw[<-] (2) edge [<-] node[right] {$14$} (4);
	\draw[<-] (2) edge [bend right] node[right] {$15$} (4);
	\draw[->] (4) to [out=360,in=20,looseness=18] (4) ;
	\draw[->] (4) to [out=330,in=350,looseness=18] (4) ;
	\draw[->] (4) to [out=300,in=320,looseness=18] (4) ;
	\draw[->] (4) to [out=270,in=290,looseness=18] (4) ;
	\draw[->] (3) edge [out=300,in=240] node[above] {$18$} (4);
	\draw[->] (3) edge [bend left] node[above] {$19$} (4);
	\draw[<-] (3) edge [<-] node[above] {$20$} (4);
	\draw[<-] (3) edge [bend right] node[above] {$21$} (4);
	\end{tikzpicture}  
	\caption{The SHS Markov Chain for system with two parallel M/M/1/1 queues 
		with multiple sources and losses of packets in service.}
	\label{fig:server-routing}
\end{figure}

\begin{table}[t!]
	\centering
		\begin{tabular}{l l l l l}
			$l$ & $q_l\to q_{l^\prime}$ & $\lambda^{l}$ & $x^\prime=x A_l$ & $\bar v_{q_l}A_l$\\\hline
			0 & $00\to 01$ & $\lambda_1p_{12}$ & $[x_0\ 0\ 0]$ & $[v_{00}(0)\ 0\ 0]$\\
			1 & $00\to 01$ & $\sum_{k>1}\lambda_kp_{k2}$ & $[x_0\ 0\ x_0 ]$ & $[v_{00}(0)\ 0\ v_{00}(0)]$\\
			2 & $01\to 00$ & $\mu_2$ & $[x_2\ 0\ 0]$ & $[v_{01}(2)\ 0\ 0]$\\
			3 & $01\to 00$ & $\theta_2$ & $[x_0\ 0\ 0]$ & $[v_{01}(0)\ 0\ 0]$\\
			4 & $01\to 01$ & $\lambda_1p_{12}$ & $[x_0\ 0\ 0]$ & $[v_{01}(0)\ 0\ 0]$\\
			5 & $01\to 01$ & $\sum_{k>1}\lambda_kp_{k2}$ & $[x_0\ 0\ x_0]$ & $[v_{01}(0)\ 0\ v_{01}(0)]$\\
			6 & $00\to 10$ & $\lambda_1p_{11}$ & $[x_0\ 0\ 0]$ & $[v_{00}(0)\ 0\ 0]$\\
			7 & $00\to 10$ & $\sum_{k>1}\lambda_kp_{k1}$ & $[x_0\ x_0\ 0 ]$ & $[v_{00}(0)\ v_{00}(0)\ 0]$\\
			8 & $10\to 00$ & $\mu_1$ & $[x_1\ 0\ 0]$ & $[v_{10}(1)\ 0\ 0]$\\
			9 & $10\to 00$ & $\theta_1$ & $[x_0\ 0\ 0]$ & $[v_{10}(0)\ 0\ 0]$\\
			10 & $10\to 10$ & $\lambda_1p_{11}$ & $[x_0\ 0\ 0]$ & $[v_{10}(0)\ 0\ 0]$\\
			11 & $10\to 10$ & $\sum_{k>1}\lambda_kp_{k1}$ & $[x_0\ x_0\ 0]$ & $[v_{10}(0)\ v_{10}(0)\ 0]$\\
			12 & $01\to 11$ & $\lambda_1p_{11}$ & $[x_0\ 0\ x_2]$ & $[v_{01}(0)\ 0\ v_{01}(2)]$\\
			13 & $01\to 11$ & $\sum_{k>1}\lambda_kp_{k1}$ & $[x_0\ x_0\ x_2 ]$ & $[v_{01}(0)\ v_{01}(0)\ v_{01}(2)]$\\
			14 & $11\to 01$ & $\mu_1$ & $[x_1\ 0\ x_2]$ & $[v_{11}(1)\ 0\ v_{11}(2)]$\\
			15 & $11\to 01$ & $\theta_1$ & $[x_0\ 0\ x_2]$ & $[v_{11}(0)\ 0\ v_{11}(2)]$\\
			16 & $11\to 11$ & $\lambda_1p_{11}$ & $[x_0\ 0\ x_2]$ & $[v_{11}(0)\ 0\ v_{11}(2)]$\\
			17 & $11\to 11$ & $\sum_{k>1}\lambda_kp_{k1}$ & $[x_0\ x_0\ x_2]$ & $[v_{11}(0)\ v_{11}(0)\ v_{11}(2)]$\\
			18 & $10\to 11$ & $\lambda_1p_{12}$ & $[x_0\ x_1\ 0]$ & $[v_{10}(0)\ v_{10}(1)\ 0]$\\
			19 & $10\to 11$ & $\sum_{k>1}\lambda_kp_{k2}$ & $[x_0\ x_1\ x_0 ]$ & $[v_{10}(0)\ v_{10}(1)\ v_{10}(0)]$\\
			20 & $11\to 10$ & $\mu_2$ & $[x_2\ x_1\ 0]$ & $[v_{11}(2)\ v_{11}(1)\ 0]$\\
			21 & $11\to 10$ & $\theta_2$ & $[x_0\ x_1\ 0]$ & $[v_{11}(0)\ v_{11}(1)\ 0]$\\
			22 & $11\to 11$ & $\lambda_1p_{12}$ & $[x_0\ x_1\ 0]$ & $[v_{11}(0)\ v_{11}(1)\ 0]$\\
			23 & $11\to 11$ & $\sum_{k>1}\lambda_kp_{k2}$ & $[x_0\ x_1\ x_0]$ & $[v_{11}(0)\ v_{11}(1)\ v_{11}(0)]$\\
	\end{tabular}
	\caption{Table of transitions of Figure~\ref{fig:server-routing}.}
	\label{tab:server-routing}
\end{table}

We now explain each transition $l$:
\begin{itemize}
	\item[]
		\begin{itemize}
			\item[$l=0$] 
			There is an arrival of source $1$ to queue 2, when the system is empty. Therefore,
			$x_0$ and $x_1$ do not change and the age of the update in service in queue 2 is zero due to a fresh 
			update arrival.
			\item[$l=1$] There is an arrival of an update of the other sources to queue 2 when it is idle. 
			Therefore, we set $x_2^\prime=x_0$, which means that, when this update ends its service, 
			the age of the monitor is again $x_0$.
			\item[$l=2$] The update under execution in queue 2 is delivered and the age of the monitor is
			updated by that of this update, i.e., $x_0^\prime=x_2$.
			\item[$l=3$] The update that is in service in queue 2 is lost and, therefore, 
			the age of the monitor does not change and queue 2 has no updates.
			\item[$l=4$] There is an arrival of source $1$ to queue 2 when it is in service. For this case,
			since preemption is permitted, the update in service is replaced by the fresh one and, 
			therefore, the age of the update in service in queue 2 changes to zero.
			\item[$l=5$] There is an arrival of another source to queue 2 when it has an update in service. 
			For this case, the update in service is replaced by the fresh one and the age of the update in service in queue 1
			changes to that of the monitor, i.e., $x_2^\prime=x_0$.
	\end{itemize}
	\item[] The transitions 6-11 are symmetric to 0-5, respectively.
	\begin{itemize}
		\item[$l=12$] When there is an update in queue 2, if an update of source $1$ arrives to queue 1, 
			the age of the monitor and of the update in queue 2 do not change, whereas that of queue 1 
			is set to zero, that is, $x_1^\prime=0$.
		\item[$l=13$] When there is an update in queue 2, if there is an arrival of one of the other
		sources in queue 1, the age of the update in queue 1 changes to the age of the monitor, i.e,. 
		$x_1^\prime=x_0$, whereas $x_0$ and $x_2$ do not change.
		\item[$l=14$] When there are updates in both queues, an update of queue 1 is delivered and
		the age of the monitor changes to $x_1$, i.e., $x_0^\prime=x_1$.
		\item[$l=15$] When there are updates in both queues, if an update of queue 1 is lost,
		the age of the monitor does not change and queue 1 is idle.
		\item[$l=16$] When both queues have updates in service, if an update of source $1$ arrives 
		to queue 1, we set $x_1^\prime$ to zero and the rest does not change.
		\item[$l=17$] When both queues have updates in service, if an update of the other sources 
		arrives to queue 1, we set $x_1^\prime$ to the same as the monitor.
	\end{itemize}
	\item[] The transitions 18-23 are symmetric to 12-17, respectively.
\end{itemize}

	The stationary distribution of the Markov Chain in  Figure~\ref{fig:server-routing} is given by
	$$
	\pi_{k_1k_2}=\frac{\rho_1^{k_1}\rho_2^{k_2}}{(1+\rho_1)(1+\rho_2)}, \ \text{ for } k_1,k_2=0,1,
	$$
	where $\rho_j=\frac{\lambda_1p_{1j}+\sum_{k>1}\lambda_kp_{kj}}{\mu_j+\theta_j}, j=1,2$. Moreover, we define 
the value of $b_q$ for each state $q\in\{00,10,01,11\}$ as follows:
$\mathbf b_{00}=[1\ 0 \ 0]$, $\mathbf b_{10}=[1\ 1 \ 0]$, $\mathbf b_{01}=[1\ 0 \ 1]$ and 
$\mathbf b_{11}=[1\ 1 \ 1]$. We also
define, for $q\in\{00,10,01,11\}$, the following vector $\mathbf v_q=[v_q(0)\ v_q(1)\ v_q(2)]$.

Let $ \hat \mu=\mu_1+\mu_2$ and $\hat \theta=\theta_1+\theta_2$. 
The SHS method says that the average AoI of this system is $\sum_q v_q(0)$,
where $v_{q}(0)$ is the solution of the following system of equations:
	\begin{align}
	\lambda\mathbf v_{00}&=[\pi_{00}\ 0 \ 0]+\mu_1 [v_{10}(1)\ 0 \ 0]\nonumber\\&+\theta_1 [v_{10}(0)\ 0 \ 0]\nonumber\\
	&+\mu_2 [v_{01}(2)\ 0 \ 0]+\theta_2 [v_{01}(0)\ 0 \ 0]\label{eq:server-routing-shs-1}
	\end{align} 
	\begin{align}
	(\lambda+\mu_1+\theta_1)\mathbf v_{10}&=[\pi_{10}\ \pi_{10}\ 0]+\lambda_1p_{11}[v_{00}(0)\ 0 \ 0]\nonumber\\ &+
	\sum_{k>1}\lambda_kp_{1k}[v_{00}(0)\ v_{00}(0) \ 0]\nonumber\\ &+\mu_2 [v_{11}(2)\ v_{11}(1) \ 0]\nonumber\\
	&+\theta_2 [v_{11}(0)\ v_{11}(1) \ 0]+\lambda_1p_{11}[v_{10}(0)\ 0 \ 0]\nonumber\\\ &+
	\sum_{k>1}\lambda_kp_{1k}[v_{10}(0)\ v_{10}(0) \ 0]
	\end{align}
	\begin{align}
	(\lambda+\mu_2+\theta_2)\mathbf v_{01}&=[\pi_{01}\ 0 \ \pi_{01}]+\lambda_1p_{12}[v_{00}(0)\ 0 \ 0]\nonumber\\\ &+
	\sum_{k>1}\lambda_kp_{k2}[v_{00}(0) \ 0 \ v_{00}(0) ]\nonumber\\\ &+\mu_1 [v_{11}(1)\ 0 \ v_{11}(2)]\nonumber\\
	&+\theta_1 [v_{11}(0)\ 0 \ v_{11}(2)]\nonumber\\\ &+\lambda_1p_{12}[v_{01}(0)\ 0 \ 0]\nonumber\\\ &+
	\sum_{k>1}\lambda_kp_{k2}[v_{01}(0)\ 0\  v_{01}(0)]
	\end{align}
	\begin{align}
	(\lambda+\hat \mu+\hat \theta)\mathbf v_{11}&=[\pi_{11}\ \pi_{11} \ \pi_{11}]+\lambda_1p_{11}[v_{01}(0)\ 0 \ v_{01}(2)]\nonumber\\\ &+
	\sum_{k>1}\lambda_kp_{k1}[v_{01}(0) \ v_{01}(0) \ v_{01}(2) ]\nonumber\\&+\lambda_1p_{11}[v_{11}(0)\ 0 \ v_{11}(2)]\nonumber\\\ &+
	\sum_{k>1}\lambda_kp_{k1}[v_{11}(0) \ v_{11}(0) \ v_{11}(2) ]\nonumber\\&+\lambda_1p_{12}[v_{10}(0)\ v_{10}(1)\ 0]\nonumber\\\ &+
	\sum_{k>1}\lambda_kp_{k2}[v_{10}(0) \ v_{10}(1) \ v_{10}(0) ]\nonumber\\&+\lambda_1p_{12}[v_{11}(0)\ v_{11}(1)\ 0]\nonumber\\\ &+
	\sum_{k>1}\lambda_kp_{k2}[v_{11}(0) \ v_{11}(1) \ v_{11}(0) ].\label{eq:server-routing-shs-4}
	\end{align}

Since the first equation has two irrelevant variables and the second and third ones have one 
irrelevant variable, the above expression can be written alternatively as a system of 8 equations.

\begin{proposition}
	The average AoI of source $1$ in the aforementioned system is given by 
	$v_{00}(0)+v_{10}(0)+v_{01}(0)+v_{11}(0)$, where for $q\in\{00,10,01,11\}$, $v_{q}(0)$ is 
	the solution of \eqref{eq:server-routing-shs-1}-\eqref{eq:server-routing-shs-4}.
\end{proposition}

\subsubsection{Average AoI Comparison}
\label{sec:real-aoi:sub:nobuffer:comparison}

\begin{figure*}
\centering
\begin{subfigure}[b]{0.3\textwidth}
         \centering
         \includegraphics[width=\columnwidth,clip=true,trim=140pt 610pt 190pt 130pt]{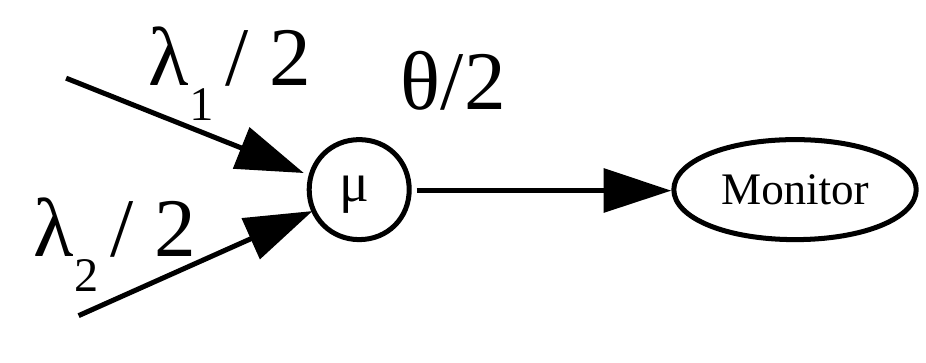}
         \caption{M/M/1/1 queue with arrival rate $\lambda/2$.}
         \label{fig:ex_nobuffer:sub1}
     \end{subfigure}
\hfill
\begin{subfigure}[b]{0.3\textwidth}
         \centering
         \includegraphics[width=\columnwidth,clip=true,trim=140pt 710pt 190pt 30pt]{fig1.pdf}
         \caption{M/M/1/1 queue with service rate $2\mu$.}
         \label{fig:ex_nobuffer:sub2}
     \end{subfigure}
\hfill
\begin{subfigure}[b]{0.3\textwidth}
         \centering
         \includegraphics[width=\columnwidth,clip=true,trim=70pt 460pt 190pt 230pt]{fig1.pdf}
         \caption{Two parallel M/M/1/1 queues.}
         \label{fig:ex_nobuffer:sub3}
     \end{subfigure}
\caption{Representation of the models under comparison in Section~\ref{sec:real-aoi:sub:nobuffer:comparison} for two sources.}
\label{fig:ex_nobuffer}
\end{figure*}

\begin{figure}
	\centering
	\includegraphics[width=\columnwidth,clip=true,trim=10pt 230pt 0pt 240pt]{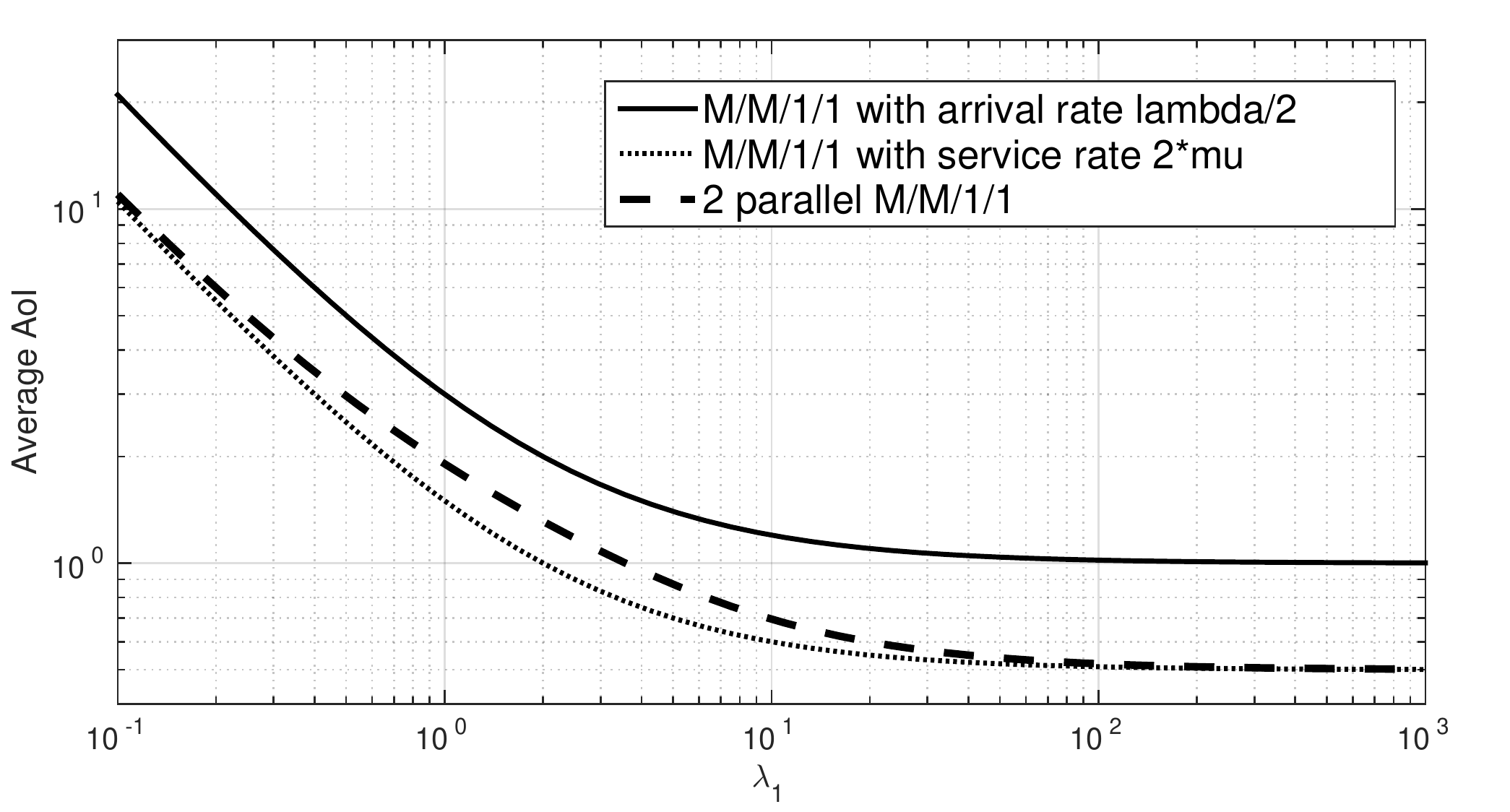}
			\caption{Average AoI of source $1$ when $\lambda_1$ varies from $0.1$ to $10^{3}$ 
			with a single source and without losses ($\lambda_k=0$ for all $k>1$ and $\theta=0$). $\mu=1$.}
	\label{fig:comparison_mm11_theta0_alpha0}
\end{figure}

\begin{figure}[t!]
	\centering
	\includegraphics[width=\columnwidth,clip=true,trim=10pt 230pt 0pt 240pt]{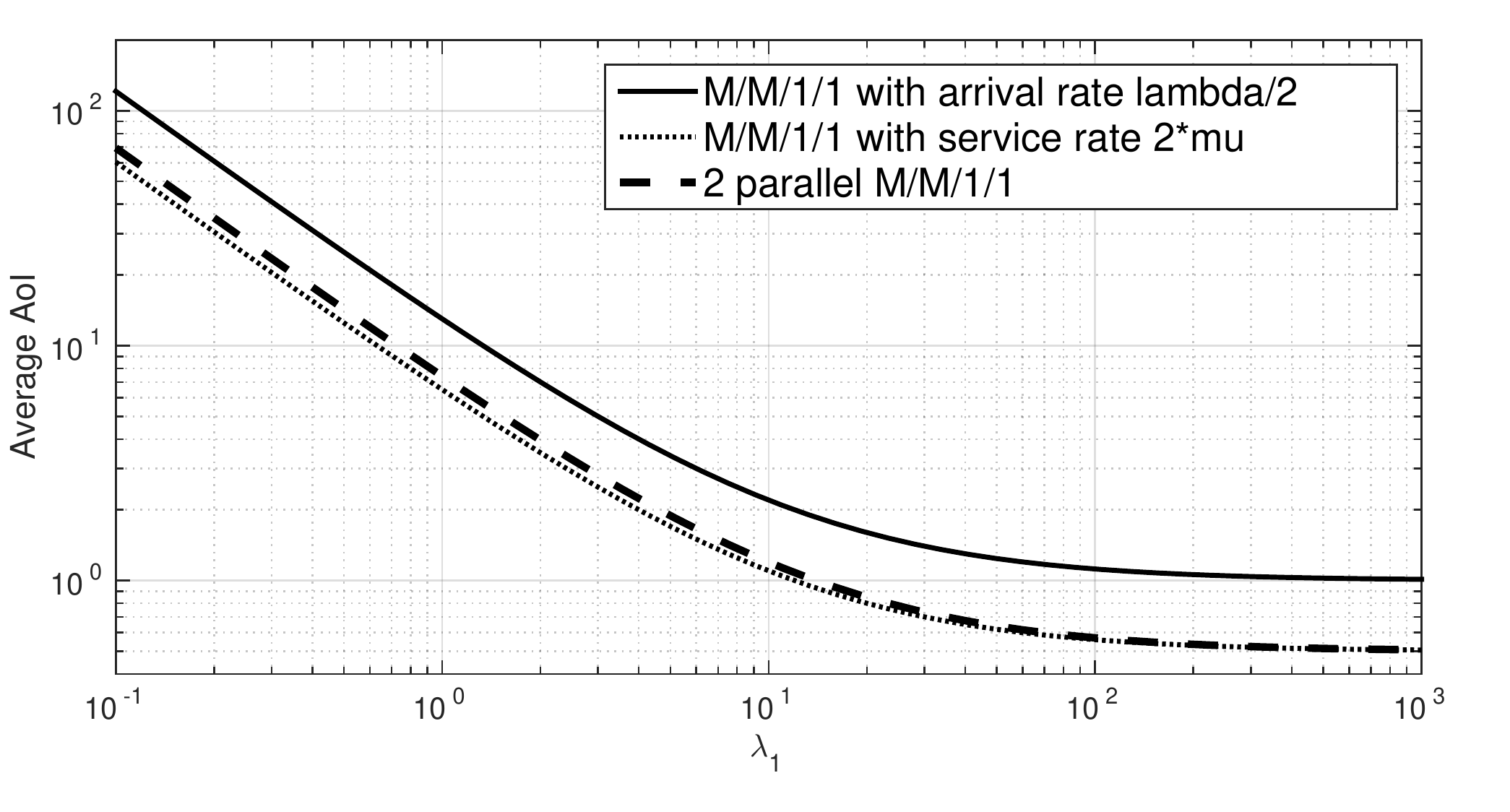}
	\caption{
			Average AoI  of source $1$  when $\lambda_1$ varies from $0.1$ to $10^{3}$ 
			with multiple sources and without losses ($\sum_{k>1}\lambda_k=10$ and $\theta=0$). $\mu=1$.}
	\label{fig:comparison_mm11_theta0_alpha10}
\end{figure}

\begin{figure}[t!]
	\centering
	\includegraphics[width=\columnwidth,clip=true,trim=10pt 230pt 0pt 240pt]{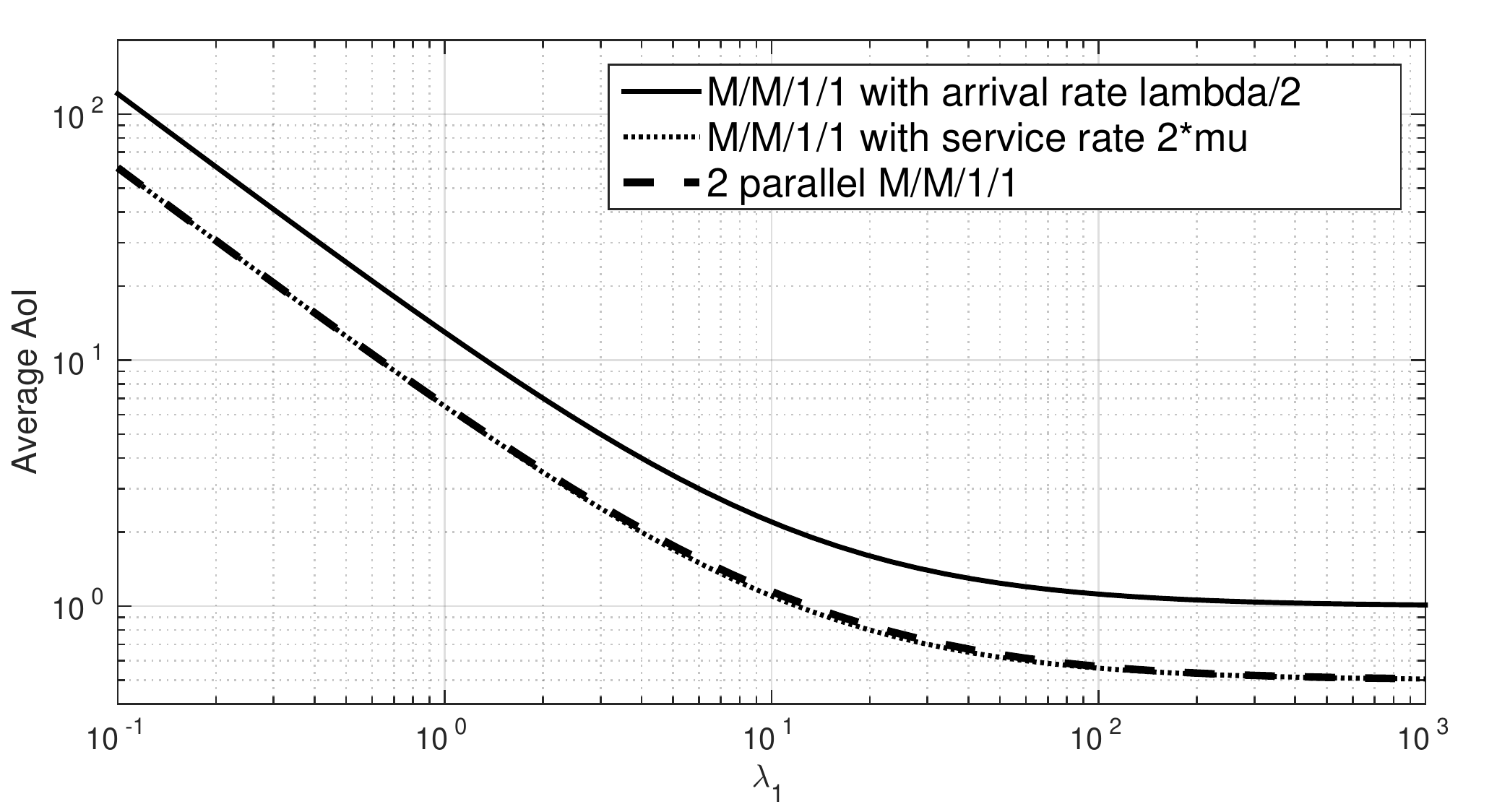}
	\caption{
			Average AoI  of source $1$  when $\lambda_1$ varies from $0.1$ to $10^{3}$ 
			with a single source and losses ($\lambda_k=0$ for all $k>1$ and $\theta=10$). $\mu=1$.
		}
	\label{fig:comparison_mm11_theta10_alpha0}
\end{figure}

\begin{figure}[t!]
	\centering
	\includegraphics[width=\columnwidth,clip=true,trim=10pt 230pt 0pt 240pt]{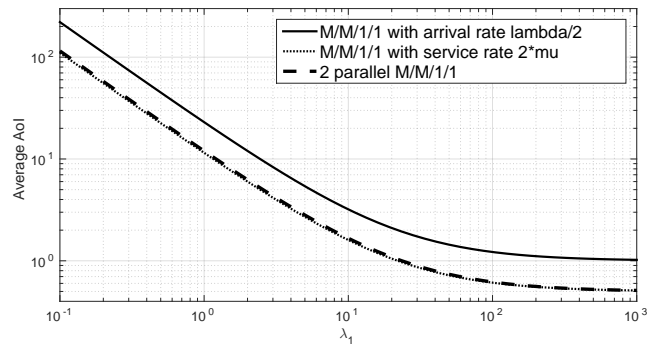}
	\caption{
			Average AoI  of source $1$  when $\lambda_1$ varies from $0.1$ to $10^{3}$ 
			with multiple sources and losses ($\sum_{k>1}\lambda_k=10$ and $\theta=10$). $\mu=1$.}
	\label{fig:comparison_mm11_theta10_alpha10}
\end{figure}


	We now compare the average AoI  of source $1$ for the models we have studied in this section.
	For this purpose, we consider three systems. The first one consists of a single M/M/1/1 queue 
	with arrival rate of source $1$ $\lambda_1/2$ and of the rest of the sources $(1/2)\sum_{k>1}\lambda_k$, loss rate $\theta/2$ and service rate $ \mu$ (see Figure~\ref{fig:ex_nobuffer:sub1}). The average 
	AoI of this model is represented with a solid line. The second system we 
	consider is a single M/M/1/1 queue with arrival rate of source $1$ $\lambda_1$ and of the rest of the sources $\sum_{k>1}\lambda_k$, loss rate $\theta$ and service rate 
	$2\mu$  (see Figure~\ref{fig:ex_nobuffer:sub2}). The average AoI of this model is represented with a dotted line. 
	Finally, we consider a system with two parallel M/M/1/1 queues with arrival rate of source $1$ equal to $\lambda_1$ and of the rest of the sources $\sum_{k>1}\lambda_k$. Besides, we consider that $p_{kj}=1/2$ for all $k=1,\dots,n$ and $j=1,2$, the loss rate and the service rate are  in both servers $\theta/2$ and $\mu$, respectively  (see Figure~\ref{fig:ex_nobuffer:sub3}). The average 
	AoI of this model is represented with a dashed line. 
	Our goal is to determine which system has the smallest average AoI when $\lambda_1$ varies and the rest of the parameters are fixed. To this end, we have solved numerically the systems of equations in 
	\eqref{eq:oneserver-shs} and in  \eqref{eq:server-routing-shs-1}-\eqref{eq:server-routing-shs-4}.
	We set $\mu=1$ in these simulations.
	When we study the system with multiple sources, we consider that $\sum_{k>1}\lambda_k=10$, and
	in the case of losses in the packets in service, we set $\theta=10$.
	\\
	In Figure~\ref{fig:comparison_mm11_theta0_alpha0}, we plot the average AoI 
	of these systems as a function of $\lambda_1$ when there is a single source and there are no 
	losses in the queues. We observe that the smallest age is achieved for the single M/M/1/1 queue system
	with service rate $2\mu$. We also observe that the age of the two parallel M/M/1/1 queues is the same 
	as the latter when $\lambda_1$ is either very small or very large.

In Figures~\ref{fig:comparison_mm11_theta0_alpha10}-\ref{fig:comparison_mm11_theta10_alpha10}, we show that the average AoI 
of a system with two parallel M/M/1/1 queues is equal to that of a single server with service rate
$2\mu$. 

In queueing theory, it is known that, among the systems under consideration in this section, 
the one that minimizes the delay is the single M/M/1/1 queue with service rate $2\mu$. 
Therefore, these illustrations show that the AoI also verifies this property. 
On the other hand, for classical queueing theory metrics such as delay, the performance
of two parallel M/M/1/1 queues coincides with that of a single M/M/1/1 queue with arrival rate $\lambda/2$ and
loss rate $\hat \theta/2$. However, as far as average AoI is concerned, one can see  that, according to the figures we present in this section,
 this is not the case (solid  and dashed lines do not coincide on these figures).

	\subsection{Queues With Buffer}
	\label{sec:real-aoi:sub:buffer}
	We now focus on queues with buffer. For this case, an  update starts getting served upon its arrival to a queue, if the queue is idle.
However, if the queue is busy, the incoming update is put in the last position of the queue
and, if the queue is full, the last update of the buffer is replaced by the new one. In this section, 
we aim to compare the average AoI of a system with one queue and buffer size two with that of  two parallel queues with buffer size one. We first compute the average AoI of the former system and then of the latter one. 

\subsubsection{The M/M/1/3* queue}

We concentrate on a system formed by a queue with a buffer of size two. 
When an update arrives and the 
system is empty, it gets served immediately. However, if an update arrives when another packet is in service, it replaces the last update in the queue if it is full and, otherwise, it is put in the last position of the queue. This system will be denoted in the remainder of 
the paper as the M/M/1/3* queue.

	When traffic comes from $n$
	different sources, we are interested in computing the average AoI of source
	$1$. Updates of source $1$ arrive to the queue according to a Poisson process of rate $\lambda_1$
	and of those of the rest of the sources with rate $\sum_{k>1}\lambda_k$.  We assume that the updates that are waiting in the queue 
are served according to the FCFS discipline and that the service time  is exponentially 
distributed with rate $\mu$, as well as the update that is in service is lost with exponentially 
distributed time with rate $\theta$.

\begin{figure}[t!]
	\centering
	\begin{tikzpicture}[]
	\node[style={circle,draw}] at (0,0) (1) {$0$};
	\node[style={circle,draw}] at (2,0) (2) {$1$};
	\node[style={circle,draw}] at (4,0) (3) {$2$};
	\node[style={circle,draw}] at (6,0) (4) {$3$};
	\node[] at (7,0) (X) {$12$};
	\node[] at (6.75,0.75) (X) {$13$};
	\draw[->] (1) edge [out=90,in=100] node[above] {$0$} (2);
	\draw[->] (1) edge [bend left] node[above] {$1$} (2);
	\draw[<-] (1) edge [bend right] node[below] {$2$} (2);
	\draw[<-] (1) edge [out=280, in=260] node[below] {$3$} (2);
	\draw[->] (4) to [out=340,in=20,looseness=8] (4) ;
	\draw[->] (4) to [out=30,in=50,looseness=18] (4) ;
	\draw[->] (2) edge [out=80,in=100] node[above] {$4$} (3);
	\draw[->] (2) edge [bend left] node[above] {$5$} (3);
	\draw[<-] (2) edge [bend right] node[below] {$6$} (3);
	\draw[<-] (2) edge [out=280, in=260] node[below] {$7$} (3);
	\draw[->] (3) edge [out=80,in=100] node[above] {$8$} (4);
	\draw[->] (3) edge [bend left] node[above] {$9$} (4);
	\draw[<-] (3) edge [bend right] node[below] {$10$} (4);
	\draw[<-] (3) edge [out=280, in=260] node[below] {$11$} (4);
	\end{tikzpicture}  
	\caption{The SHS Markov Chain for the M/M/1/3* queue with multiple sources and losses
		of packets in service.}
	\label{fig:mm13}
\end{figure}
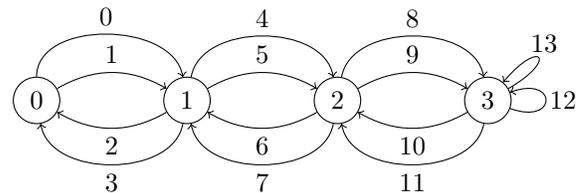

	We employ the SHS method to calculate the average AoI of this system.
	The continuous state is given by $\mathbf x(t)=[x_0(t)\ x_1(t)\ x_2(t)\ x_3(t)]$, where the correspondence between $x_i$(t) and each element in the system is as follows: $x_0$ is 
	the age of the monitor, $x_1$ is the age if the update in service is delivered and $x_2$ and $x_3$ is 
	respectively the age if the update in the first and second positions of the queue are delivered.
The discrete state is a four state Markov Chain, where state $k$ represents that there are $k$
	updates present in the system, with $k=0,1,2,3$. The Markov Chain under consideration and the SHS
	transition are represented respectively in Figure~\ref{fig:mm13} and Table~\ref{tab:mm13}.

\begin{table}[t!]
	\centering
		\begin{tabular}{l l l l l}
			$l$ & $q_l\to q_{l^\prime}$ & $\lambda^{l}$ & $x^\prime=x A_l$ & $\bar v_{q_l}A_l$\\\hline
			0 & $0\to 1$ & $\lambda_1$ & $[x_0\ 0\ 0\ 0]$ & $[v_{0}(0)\ 0\ 0\ 0]$\\
			1 & $0\to 1$ & $\sum_{k>1}\lambda_k$ & $[x_0\ x_0\ 0\ 0]$ & $[v_{0}(0)\ v_{0}(0)\ 0\ 0]$\\
			2 & $1\to 0$ & $\mu$ & $[x_1\ 0\ 0\ 0]$ & $[v_{1}(1)\ 0\ 0\ 0]$\\
			3 & $1\to 0$ & $\theta$ & $[x_0\ 0\ 0\ 0]$ & $[v_{1}(0)\ 0\ 0\ 0]$\\
			4 & $1\to 2$ & $\lambda_1$ & $[x_0\ x_1\ 0\ 0]$ & $[v_{1}(0)\ v_1(1)\ 0\ 0]$\\
			5 & $1\to 2$ & $\sum_{k>1}\lambda_k$ & $[x_0\ x_1\ x_1\ 0]$ & $[v_{1}(0)\ v_{1}(1)\ v_1(1)\ 0]$\\
			6 & $2\to 1$ & $\mu$ & $[x_1\ x_2\ 0\ 0]$ & $[v_{2}(1)\ v_2(2)\ 0\ 0]$\\
			7 & $2\to 1$ & $\theta$ & $[x_0\ x_2\ 0\ 0]$ & $[v_{2}(0)\ v_2(2)\ 0\ 0]$\\
			8 & $2\to 3$ & $\lambda_1$ & $[x_0\ x_1\ x_2\ 0]$ & $[v_{2}(0)\ v_2(1)\ v_2(2)\ 0]$\\
			9 & $2\to 3$ & $\sum_{k>1}\lambda_k$ & $[x_0\ x_1\ x_2\ x_2]$ & $[v_{2}(0)\ v_2(1)\ v_2(2)\ v_2(2)]$\\
			10 & $3\to 2$ & $\mu$ & $[x_1\ x_2\ x_3\ 0]$ & $[v_{3}(1)\ v_3(2)\ v_3(3)\ 0]$\\
			11 & $3\to 2$ & $\theta$ & $[x_0\ x_2\ x_3\ 0]$ & $[v_{3}(0)\ v_3(2)\ v_3(3)\ 0]$\\
			12 & $3\to 3$ & $\lambda_1$ & $[x_0\ x_1\ x_2\ 0]$ & $[v_{3}(0)\ v_3(1)\ v_3(2)\ 0]$\\
			13 & $3\to 3$ & $\sum_{k>1}\lambda_k$ & $[x_0\ x_1\ x_2\ x_2]$ & $[v_{3}(0)\ v_{3}(1)\ v_3(2) \ v_3(2)]$\\
	\end{tabular}
	\caption{Table of transitions of Figure~\ref{fig:mm13}}
	\label{tab:mm13}
\end{table}

We now explain each transition $l$:
\begin{itemize}
	\item[]
	\begin{itemize}
		\item[$l=0$] The system is empty and an update of source $1$ arrives. The age of the monitor is
		not modified and we set $x_1^\prime=0$.
		\item[$l=1$] The system is empty and an update of another source arrives. The age of the monitor is
		not modified and the age $x_1$ changes to $x_0$, i.e., $x_1^\prime=x_0$.
		\item[$l=2$] When there is an update getting in service and the queue is empty. If the update in 
		service is delivered, the age of the monitor changes to $x_1$, i.e., $x_0^\prime=x_1$.
		\item[$l=3$] The queue is empty and the update in service is lost. For this case, the age of the
		monitor does not change and the age of $x_1$ is replaced by zero.
		\item[$l=4$] The queue is busy and an update of source $i$ arrives. The age of the monitor and that of 
		$x_1$ are not modified and we set $x_2^\prime=0$.
		\item[$l=5$] The queue is busy and an update of source $i$ arrives. The age of the monitor and that of 
		$x_1$ are not modified and the age $x_2$ changes to $x_1$, i.e., $x_2^\prime=x_1$.
		\item[$l=6$] There are two updates in the system and the update in service is delivered and, 
		therefore, the age of the monitor changes to $x_1$ and the age $x_2$ to $x_1$, i.e., $x_0^\prime=x_1$ and $x_1^\prime=x_2$.
		\item[$l=7$] There are two updates in the system and the update in service is lost. For this case, the age of the monitor does not change, but the age $x_1$ is replaced by $x_2$, i.e., $x_1^\prime=x_2$ 
		since the update that was waiting start getting served.
		\item[$l=8$]   There are two updates in the system and an update of source $1$ arrives. 
		The ages of the updates that are present in the system do not change and we set $x_3^\prime=0$.
		\item[$l=9$] There are two updates in the system and an update of another source arrives. 
		The ages of the updates that are present in the system do not change 
		and the age $x_3$ changes to $x_2$, i.e., $x_3^\prime=x_2$.
		\item[$l=10$] The system is full and the update in service is delivered. For this case, 
		the age of the monitor changes to $x_1$, the age of $x_1$ to $x_2$ and the age of $x_2$ to $x_3$.
		\item[$l=11$] The system is full and the update in service is lost. For this case, the age of the
		monitor does not change, but the age of $x_1$ changes to $x_2$ and the age of $x_2$ to $x_3$.
		\item[$l=12$] The system is full and an update of source $1$ arrives. The age of the monitor
		and of that of $x_1$ and $x_2$ are not modified and we set $x_3^\prime=0$.
		\item[$l=13$] The system is full and an update of another source arrives. The ages of the monitor,
		of $x_1$ and of $x_2$ do not change, but the age of $x_3$ is set to $x_2$, i.e., $x_3^\prime=x_2$.
	\end{itemize}
\end{itemize}

Let $\lambda=\sum_{k=1}^n\lambda_k$ and $\rho=\frac{\lambda}{\mu+\theta}$. The stationary distribution of the Markov Chain in Figure~\ref{fig:mm13} is
$$
\pi_j=\frac{ \rho^j}{1+\rho+\rho^2+\rho^3}, \ j=0,1,2,3.
$$

We now define the vector $\mathbf b_q$ for all state $q$ of the Markov Chain of 
Figure~\ref{fig:mm13}: $\mathbf b_0=[1\ 0\ 0\ 0]$,  $\mathbf b_1=[1\ 1\ 0\ 0]$, 
$\mathbf b_2=[1\ 1\ 1\ 0]$ and  $\mathbf b_3=[1\ 1\ 1\ 1]$. Besides, for all state $q\in\{0,1,2,3\}$, 
$\mathbf v_q=[v_q(0)\ v_q(1)\ v_q(2)\ v_q(3)]$. From Theorem 4 in \cite{YK19}, we have that
the average AoI of the M/M/1/3* queue is $v_0(0)+v_1(0)+v_2(0)+v_3(0)$, where
$v_q(0)$ is the solution of the following system of equations:
	\begin{align}
	\lambda \mathbf v_0=&[\pi_0\ 0\ 0 \ 0]+\mu[v_1(1)\ 0\ 0\ 0]+\theta[v_1(0)\ 0\ 0\ 0]\label{eq:mm13-shs-1}\\
	(\lambda+\mu+\theta)\mathbf v_1=&[\pi_1\ \pi_1\ 0 \ 0]+\lambda_1[v_0(0)\ 0\ 0\ 0]\nonumber\\
	&+\sum_{k>1}\lambda_k[v_0(0)\ v_0(0)\ 0\ 0]\nonumber\\
	&+\mu[v_2(1)\ v_2(2)\ 0\ 0]\nonumber\\
	&+\theta[v_2(0)\ v_2(2)\ 0\ 0].\\
	(\lambda+\mu+\theta)\mathbf v_2=&[\pi_2\ \pi_2\ \pi_2 \ 0]+\lambda_1[v_1(0)\ v_1(1)\ 0\ 0]\nonumber\\
	&+\sum_{k>1}\lambda_k[v_1(0)\ v_1(1)\ v_1(1)\ 0]\nonumber\\
	&+\mu[v_3(1)\ v_3(2)\ v_3(3)\ 0]\nonumber\\
	&+\theta[v_3(0)\ v_3(2)\ v_3(3)\ 0].\\
	(\lambda+\mu+\theta)\mathbf v_3=&[\pi_3\ \pi_3\ \pi_3 \ \pi_3]+\lambda_1[v_2(0)\ v_2(1)\ v_2(2)\ 0]\nonumber\\
	&+\sum_{k>1}\lambda_k[v_2(0)\ v_2(1)\ v_2(2)\ v_2(2)]\nonumber\\
	&+\lambda_1[v_3(0)\ v_3(1)\ v_3(2)\ 0]\nonumber\\
	&+\sum_{k>1}\lambda_k[v_3(0)\ v_3(1)\ v_3(2)\ v_3(2)]\label{eq:mm13-shs-4}
	\end{align}

Since the first equation has three irrelevant variables and the second and third equations have
respectively two  and one irrelevant variables, the above expression can be  alternatively written as a system of 10 equations.
\begin{proposition}
	The average AoI of source $1$ in the aforementioned system is given by 
	$v_{0}(0)+v_{1}(0)+v_{2}(0)+v_{3}(0)$, where for $q\in\{0,1,2,3\}$, $v_{q}(0)$ is 
	the solution of \eqref{eq:mm13-shs-1}-\eqref{eq:mm13-shs-4}.
\end{proposition}

\subsubsection{Two parallel M/M/1/2* queues}
\label{sec:mm12-routing}

	We consider a system formed by two parallel queues with buffer size equal to one and $n$ different sources. 
	The packets are dispatched
	according to a predefined probabilistic routing. An update of source $1$ arrives to the system with rate $\lambda_1$ and it is sent to queue $j$ with probability $p_{1j}$. Therefore, the arrival rate of source $1$ to queue $j$ is $\lambda_1p_{1j}$. Besides, an update of the rest of the sources 
	arrives to the system with rate $\sum_{k>1}\lambda_k$ and an update of source $k\neq i$ is sent to queue $j$ with probability $p_{kj}$. Therefore, the arrival rate of the rest of the sources to queue $j$ is $\sum_{k>1}\lambda_kp_{kj}$. 
	If an update finds the queue full, it replaces the update that is waiting
	in the queue, whereas when the queue is idle, it starts being served immediately. This system
	will be denoted as two parallel M/M/1/2* queues.

	We assume that the service time of queue $j$ is exponentially distributed with rate
	$\mu_j$ and that updates/packets in service in queue $j$ are lost with exponential time of rate 
	$\theta_j$, $j=1,2$. We assume that the queues are decentralized
	in the sense that they do not communicate between each other.


	We seek to compute the average AoI of this system using the SHS method. The 
	continuous state is $\mathbf x(t)=[x_0(t)\ x_{11}(t)\ x_{12}(t)\ x_{21}(t)\ x_{22}(t)]$, where 
	the correspondence between $x_i$(t) and each element is as follows: $x_0$
	is the age of the monitor, $x_{j1}$ is the age if the update/packet in service in queue $j$ is delivered and
	$x_{j2}$ the age if the update that is waiting for service in queue $j$ is delivered. The discrete
	state is described by a Markov Chain, where the state $k_1k_2$ denotes that there are $k_1$ 
	updates in queue 1 and $k_2$ in queue 2, with   $k\in\{0,1,2\}$. The Markov Chain we study
	is depicted in Figure~\ref{fig:mm12-routing}. We note that some of the links are unified to 
	avoid heavy notation. The SHS transitions for this model are reported in Appendix~\ref{app:table-mm12-routing}.

\begin{figure}[t!]
	\centering
	\begin{tikzpicture}[]
	\node[style={circle,draw}] at (0,0) (1) {$00$};
	\node[style={circle,draw}] at (2,0) (2) {$01$};
	\node[style={circle,draw}] at (4,0) (22) {$02$};
	\node[style={circle,draw}] at (0,-2) (3) {$10$};
	\node[style={circle,draw}] at (0,-4) (32) {$20$};
	\node[style={circle,draw}] at (2,-2) (4) {$11$};
	\node[style={circle,draw}] at (2,-4) (42) {$21$};
	\node[style={circle,draw}] at (4,-2) (43) {$12$};
	\node[style={circle,draw}] at (4,-4) (5) {$22$};
	\node[] at (5.25,0) (X1) {$4$};
	\node[] at (5.25,-2) (X1) {$15$};
	\node[] at (-0.75,-4.5) (X1) {$22$};
	\node[] at (2,-5) (X1) {$25$};
	\node[] at (4,-5) (X1) {$28$};
	\node[] at (5.25,-4) (X1) {$29$};
	\draw[->] (1) edge [bend left] node[above] {$0$} (2);
	\draw[<-] (1) edge [bend right] node[above] {$1$} (2);
	\draw[->] (2) edge [bend left] node[above] {$2$} (22);
	\draw[<-] (2) edge [bend right] node[above] {$3$} (22);
	\draw[->] (1) edge [bend right] node[left] {$5$} (3);
	\draw[<-] (1) edge [bend left] node[left] {$6$} (3);
	\draw[->] (2) edge [bend right] node[right] {$7$} (4);
	\draw[<-] (2) edge [bend left] node[right] {$8$} (4);
	\draw[->] (22) edge [bend right] node[right] {$9$} (43);
	\draw[<-] (22) edge [bend left] node[right] {$10$} (43);
	\draw[->] (3) edge [bend left] node[above] {$11$} (4);
	\draw[<-] (3) edge [bend right] node[above] {$12$} (4);
	\draw[<-] (4) edge [bend left] node[above] {$13$} (43);
	\draw[->] (4) edge [bend right] node[above] {$14$} (43);
	\draw[->] (3) edge [bend right] node[left] {$16$} (32);
	\draw[<-] (3) edge [bend left] node[left] {$17$} (32);
	\draw[->] (4) edge [bend right] node[right] {$18$} (42);
	\draw[<-] (4) edge [bend left] node[right] {$19$} (42);
	\draw[->] (43) edge [bend right] node[right] {$20$} (5);
	\draw[<-] (43) edge [bend left] node[right] {$21$} (5);
	\draw[->] (32) edge [bend right] node[above] {$24$} (42);
	\draw[<-] (32) edge [bend left] node[above] {$23$} (42);
	\draw[->] (42) edge [bend right] node[above] {$27$} (5);
	\draw[<-] (42) edge [bend left] node[above] {$26$} (5);
	\draw[->] (32) to [out=210,in=230,looseness=8] (32) ;
	\draw[->] (5) to [out=260,in=280,looseness=8] (5) ;
	\draw[->] (42) to [out=260,in=280,looseness=8] (42) ;
	\draw[->] (43) to [out=340,in=20,looseness=8] (43) ;
	\draw[->] (22) to [out=340,in=20,looseness=8] (22) ;
	\draw[->] (5) to [out=340,in=20,looseness=8] (5) ;
	\end{tikzpicture}  
	\caption{The SHS Markov Chain for system with two parallel M/M/1/2* queues 
		with multiple sources and losses of packets in service.}
	\label{fig:mm12-routing}
\end{figure}
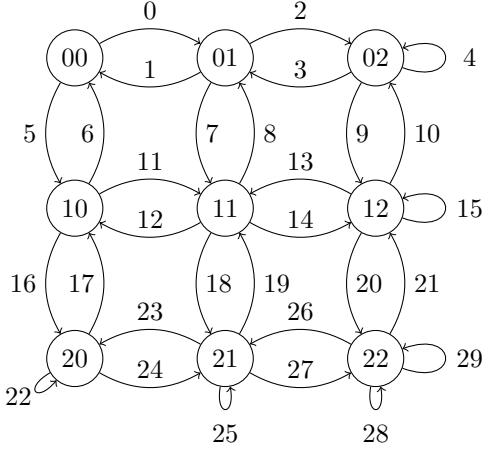

	Let $\rho_1=\frac{\lambda_1p_{11}+\sum_{k>1}\lambda_kp_{k1}}{\mu_1+\theta_1}$ and $\rho_2=\frac{\lambda_1p_{12}+\lambda_kp_{k2}}{\mu_2+\theta_2}.$ The stationary distribution of the Markov Chain in Figure~\ref{fig:mm12-routing} is 
	$$
	\pi_{k_1k_2}=\frac{\rho_1^{k_1}\rho_2^{k_2}}{(1+\rho_1+\rho_1^2)(1+\rho_2+\rho_2^2)}, k_1,k_2\in\{0,1,2\}.
	$$

Let $\hat \mu=\mu_1+\mu_2$, $\hat \theta=\theta_1+\theta_2$ and $\mathcal Q=\{00,10,20,01,11,21,02,12,22\}$. For every $q\in \mathcal Q$, we define 
$\mathbf v_q=[v_q(0)\ v_q(1)\ v_q(2)\ v_q(3)\ v_q(4)]$ and the vector $\mathbf b_{q}$ as follows: 
$\mathbf b_{00}=[1\ 0\ 0\ 0\ 0]$, $\mathbf b_{10}=[1\ 1\ 0\ 0\ 0]$, $\mathbf b_{20}=[1\ 1\ 1\ 0\ 0]$,
$\mathbf b_{01}=[1\ 0\ 0\ 1\ 0]$, $\mathbf b_{11}=[1\ 1\ 0\ 1\ 0]$, $\mathbf b_{21}=[1\ 1\ 1\ 1\ 0]$,
$\mathbf b_{02}=[1\ 0\ 0\ 1\ 1]$, $\mathbf b_{12}=[1\ 1\ 0\ 1\ 1]$ and $\mathbf b_{22}=[1\ 1\ 1\ 1\ 1]$. We use the result of Theorem 4 in \cite{YK19} that shows that the average AoI
of this system is given by $\sum_{q\in \mathcal Q}v_q(0),$ where $v_q(0)$ is the solution to the
following system of equations:
	\begin{align}
	\lambda \mathbf v_{00}=&[\pi_{00}\ 0\ 0\ 0\ 0]+\mu_1[v_{10}(1)\ 0 \ 0\ 0\ 0]\nonumber\\
	&+\theta_1[v_{10}(0)\ 0 \ 0\ 0\ 0]+\mu_2[v_{01}(3)\ 0 \ 0\ 0\ 0]\nonumber\\
	&+\theta_2[v_{01}(0)\ 0 \ 0\ 0\ 0]\label{eq:mm12-routing-shs-1}
	\end{align}
	\begin{align}
	(\lambda+\mu_1+\theta_1) \mathbf v_{10}=&[\pi_{10}\ \pi_{10}\ 0\ 0\ 0]\nonumber\\&+\lambda_1p_{11}
	[v_{00}(0)\ 0\ 0\ 0\ 0]\nonumber\\
	&+\sum_{k>1}\lambda_kp_{k1}[v_{00}(0)\  v_{00}(0)\ 0 \ 0 \ 0]\nonumber\\
	&+\mu_2[v_{11}(3)\ v_{11}(1) \ 0\ 0\ 0]\nonumber\\
	&+\theta_2[v_{11}(0)\ v_{11}(1) \ 0\ 0\ 0]\nonumber\\
	&+\mu_1[v_{20}(1)\ v_{20}(2) \ 0\ 0\ 0]\nonumber\\
	&+\theta_1[v_{20}(0)\ v_{20}(2) \ 0\ 0\ 0]\nonumber\\
	&+\mu_2[v_{11}(3)\ v_{11}(1) \ 0\ 0\ 0]\nonumber\\
	&+\theta_2[v_{11}(0)\ v_{11}(1) \ 0\ 0\ 0]
	\end{align}
	\begin{align}
	(\lambda+\mu_1+\theta_1) \mathbf v_{20}=&[\pi_{20}\ \pi_{20}\ \pi_{20}\ 0\ 0]\nonumber\\
	&+\mu_2[v_{21}(3)\ v_{21}(1) \ v_{21}(2)\ 0\ 0]\nonumber\\
	&+\theta_2[v_{21}(0)\ v_{21}(1) \ v_{21}(2)\ 0\ 0]\nonumber\\
	&+\lambda_1p_{11}
	[v_{10}(0)\ v_{10}(1)\ 0\ 0\ 0]\nonumber\\
	&+\sum_{k>1}\lambda_kp_{k1}[v_{10}(0)\  v_{10}(1)\ v_{10}(1) \ 0 \ 0]\nonumber\\
	&+\lambda_1p_{11}
	[v_{20}(0)\ v_{20}(1)\ 0\ 0\ 0]\nonumber\\
	&+\sum_{k>1}\lambda_kp_{k1}[v_{20}(0)\  v_{20}(1)\ v_{20}(1) \ 0 \ 0]
	\end{align}
	\begin{align}
	(\lambda+\mu_2+\theta_2) \mathbf v_{01}=&[\pi_{01}\ 0 \ 0 \pi_{01}\ 0]\nonumber\\
	&+\lambda_1p_{12}[v_{00}(0)\ 0\ 0\ 0\ 0]\nonumber\\
	&+\sum_{k>1}\lambda_kp_{k2}[v_{00}(0)\ 0 \ 0  v_{00}(0) \ 0]\nonumber\\
	&+\mu_2[v_{02}(3)\ 0 \ 0 \ v_{20}(4) \ 0]\nonumber\\
	&+\theta_2[v_{02}(0)\ 0 \ 0 \ v_{20}(4) \ 0]\nonumber\\
	&+\mu_1[v_{11}(1)\ 0 \ 0\ v_{11}(3) \ 0]\nonumber\\
	&+\theta_1[v_{11}(0)\ 0 \ 0 \ v_{11}(3)\ 0]
	\end{align}
	\begin{align}
	(\lambda+\hat \mu+\hat \theta) \mathbf v_{11}=&[\pi_{11}\ \pi_{11} \ 0\ \pi_{11}\ 0]\nonumber\\
	&+\lambda_1p_{11}
	[v_{01}(0)\ 0\ 0 \ v_{01}(3)\ 0]\nonumber\\
	&+\sum_{k>1}\lambda_kp_{k1}[v_{01}(0)\ v_{01}(0) \ 0\   v_{01}(3) \ 0]\nonumber\\
	&+\lambda_1p_{12}[v_{10}(0)\ v_{10}(1) \ 0 \ 0 \  0]\nonumber\\
	&+\sum_{k>1}\lambda_kp_{k2}[v_{10}(0)\ v_{10}(1) \ 0  v_{10}(0) \ 0]\nonumber\\
	&+\mu_2[v_{12}(3)\ v_{12}(1) \ 0 \ v_{21}(4) \ 0]\nonumber\\
	&+\theta_2[v_{12}(0)\ v_{12}(1) \ 0 \ v_{21}(4) \ 0]\nonumber\\
	&+\mu_1[v_{21}(1)\ v_{21}(2) \ 0\ v_{21}(3) \ 0]\nonumber\\
	&+\theta_1[v_{21}(0)\ v_{21}(2) \ 0 \ v_{21}(3)\ 0]
	\end{align}
	\begin{align}
	(\lambda+\hat \mu+\hat \theta) \mathbf v_{21}=&[\pi_{21}\ \pi_{21} \ \pi_{21} \pi_{21}\ 0]\nonumber\\
	+&\lambda_1p_{11}[v_{11}(0)\ v_{11}(1)\ 0 \ v_{11}(3)\ 0]\nonumber\\
	+&\sum_{k>1}\lambda_kp_{k1}[v_{11}(0)\ v_{11}(1) \ v_{11}(1)  \ v_{11}(3) \ 0]\nonumber\\
	+&\lambda_1p_{11}[v_{21}(0)\ v_{21}(1)\ 0 \ v_{21}(3)\ 0]\nonumber\\
	+&\sum_{k>1}\lambda_kp_{k1}[v_{21}(0)\ v_{21}(1) \ v_{21}(1) \  v_{21}(3) \ 0]\nonumber\\
	+&\lambda_1p_{12}[v_{20}(0)\ v_{20}(1) \ v_{20}(2) \ 0 \  0]\nonumber\\
	+&\sum_{k>1}\lambda_kp_{k2}[v_{20}(0)\ v_{20}(1) \ v_{20}(2)  v_{20}(0) \ 0]\nonumber\\
	+&\mu_2[v_{22}(3)\ v_{22}(1) \ v_{22}(2) \ v_{22}(4) \ 0]\nonumber\\
	+&\theta_2[v_{22}(0)\ v_{22}(1) \ v_{22}(2) \ v_{22}(4) \ 0]
	\end{align}
	\begin{align}
	(\lambda+\mu_2+\theta_2) \mathbf v_{02}=&[\pi_{02}\ 0 \ 0\  \pi_{02}\ \pi_{02}]\nonumber\\
	&+\lambda_1p_{12}
	[v_{01}(0)\ 0\ 0\ v_{01}(3)\ 0]\nonumber\\
	&+\sum_{k>1}\lambda_kp_{k2}[v_{01}(0)\ 0 \ 0 \  v_{01}(3) \ v_{01}(3)]\nonumber\\
	&+\lambda_1p_{12}[v_{02}(0)\ 0\ 0\ v_{02}(3)\ 0]\nonumber\\
	&+\sum_{k>1}\lambda_kp_{k2}[v_{02}(0)\ 0 \ 0\   v_{02}(3) \ v_{02}(3)]\nonumber\\
	&+\mu_1[v_{12}(1)\ 0 \ 0\ v_{12}(3) \ v_{12}(4)]\nonumber\\
	&+\theta_1[v_{12}(0)\ 0 \ 0 \ v_{12}(3)\ v_{12}(4)]
	\end{align}
	\begin{align}
	(\lambda+\hat \mu+\hat \theta) \mathbf v_{12}&=[\pi_{12}\ \pi_{12} \ 0\ \pi_{12}\ \pi_{12}]\nonumber\\
	&+\lambda_1p_{11}
	[v_{11}(0)\ v_{11}(1)\ 0 \ v_{11}(3)\ 0]\nonumber\\
	&+\sum_{k>1}\lambda_kp_{k1}[v_{11}(0)\ v_{11}(1) \ 0 \  v_{11}(3) \ v_{11}(3)]\nonumber\\
	&+\lambda_1p_{12}[v_{20}(0)\ v_{20}(1) \ v_{20}(2) \ 0 \  0]\nonumber\\
	&+\sum_{k>1}\lambda_kp_{k2}[v_{20}(0)\ v_{20}(1) \ v_{20}(2) \ v_{20}(0) \ 0]\nonumber\\
	&+\mu_1[v_{22}(1)\ v_{22}(2) \ 0\ v_{22}(3) \ v_{22}(4)]\nonumber\\
	&+\theta_1[v_{22}(0)\ v_{22}(2) \ 0 \ v_{22}(3)\ v_{22}(4)]
	\end{align}
	\begin{align}
	(\lambda+\hat \mu+\hat \theta) \mathbf v_{22}&=[\pi_{22}\ \pi_{22} \ \pi_{22}\ \pi_{22}\ \pi_{22}]\nonumber\\
	+&\lambda_1p_{11}
	[v_{12}(0)\ v_{12}(1)\ 0 \ v_{12}(3)\ v_{12}(4)]\nonumber\\
	+&\sum_{k>1}\lambda_kp_{k1}[v_{12}(0)\ v_{12}(1) \ v_{12}(1)\   v_{12}(3) \ v_{12}(4)]\nonumber\\
	+&\lambda_1p_{11}
	[v_{22}(0)\ v_{22}(1)\ 0 \ v_{22}(3)\ v_{22}(4)]\nonumber\\
	+&\sum_{k>1}\lambda_kp_{k1}[v_{22}(0)\ v_{22}(1) \ v_{22}(1)\   v_{22}(3) \ v_{22}(4)]\nonumber\\
	+&\lambda_1p_{12}[v_{21}(0)\ v_{21}(1) \ v_{21}(2) \ v_{21}(3) \  0]\nonumber\\
	+&\sum_{k>1}\lambda_kp_{k2}[v_{21}(0)\ v_{21}(1) \ v_{21}(2) \ v_{21}(3) \ v_{21}(3)]\nonumber\\
	+&\lambda_1p_{12}[v_{22}(0)\ v_{22}(1) \ v_{22}(2) \ v_{22}(3) \  0]\nonumber\\
	+&\sum_{k>1}\lambda_kp_{k2}[v_{22}(0)\ v_{22}(1) \ v_{22}(2) \ v_{22}(3) \ v_{22}(3)]\label{eq:mm12-routing-shs-9}
	\end{align}

From the above expressions, if we remove the irrelevant variables, we obtain a system of 27 equations.

\begin{proposition}
	The average AoI of source $1$ in the aforementioned system is given by 
	$\sum_{q\in\mathcal Q}v_q(0)$, where $v_{q}(0)$ is 
	the solution of \eqref{eq:mm12-routing-shs-1}-\eqref{eq:mm12-routing-shs-9}.
\end{proposition}

\subsubsection{Age Comparison}
\label{sec:real-aoi:sub:buffer:comparison}

\begin{figure*}
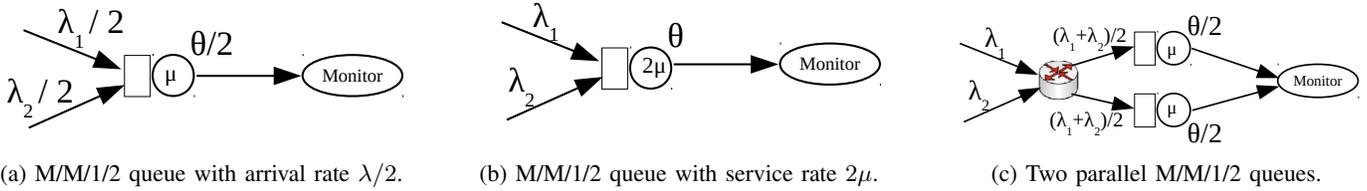

\centering
\begin{subfigure}[b]{0.3\textwidth}
         \centering
         \includegraphics[width=\columnwidth,clip=true,trim=120pt 220pt 190pt 520pt]{fig1.pdf}
         \caption{M/M/1/2 queue with arrival rate $\lambda/2$.}
         \label{fig:ex_buffer:sub1}
     \end{subfigure}
\hfill
\begin{subfigure}[b]{0.3\textwidth}
         \centering
         \includegraphics[width=\columnwidth,clip=true,trim=120pt 320pt 190pt 420pt]{fig1.pdf}
         \caption{M/M/1/2 queue with service rate $2\mu$.}
         \label{fig:ex_buffer:sub2}
     \end{subfigure}
\hfill
\begin{subfigure}[b]{0.3\textwidth}
         \centering
         \includegraphics[width=\columnwidth,clip=true,trim=50pt 90pt 190pt 640pt]{fig1.pdf}
         \caption{Two parallel M/M/1/2 queues.}
         \label{fig:ex_buffer:sub3}
     \end{subfigure}
\caption{Representation of the models under comparison in Section~\ref{sec:real-aoi:sub:buffer:comparison} for two sources.}
\label{fig:ex_buffer}
\end{figure*}

	We compare the average AoI of the models presented in this section.
	We focus on the following three systems. First, we consider an M/M/1/3* queue 
	with arrival rate of source $1$ equal to $\lambda_1/2$ and that of the rest of the sources $(1/2)\sum_{k>1}\lambda_k$, loss rate $\theta/2$ and service rate $\mu$ (see Figure~\ref{fig:ex_buffer:sub1}). The average 
	AoI of this model is represented with a solid line. The second system we 
	consider is an M/M/1/3* queue with arrival rate of source $1$ equal to $\lambda_1$ and that of the rest of the sources $\sum_{k>1}\lambda_k$, 
	loss rate $\theta$ and service rate 
	$2\mu$ (see Figure~\ref{fig:ex_buffer:sub2}). The average AoI of this model is represented with a dotted line. 
	We also consider a system with two parallel M/M/1/2* queues with arrival rate of source $1$ $\lambda_1$ and  that of the rest of the sources $\sum_{k>1}\lambda_k$. Each of the servers satisfies that $p_{k1}=p_{k2}=1/2$ for all $k=1,\dots,n$, has a loss rate equal to $ \theta/2$ and a service rate $ \mu$  (see Figure~\ref{fig:ex_buffer:sub3}). The average 
	AoI of the latter model is represented with a dashed line. 
	We aim to investigate which  system has the smallest average AoI when $\lambda_1$ varies. Thus, we have solved numerically the systems of equations in 
	\eqref{eq:mm13-shs-1}-\eqref{eq:mm13-shs-4} and of \eqref{eq:mm12-routing-shs-1}-\eqref{eq:mm12-routing-shs-9}.
	We set $\mu=1$ in these simulations.
	When we study the system with multiple sources, we consider that $\sum_{k>1}\lambda_{k}=10$ and
	in the case of packet losses, we set $\theta=10$.

We first focus on the average AoI for a single source and when there are no packet losses.
The evolution of the AoI of source $1$ with respect to $\lambda_1$ is 
represented in Figure~\ref{fig:comparison_mm12_theta0_alpha0}. We observe that for the 
M/M/1/3* queue with service rate $2\mu$, the average AoI coincides
with that of the two parallel M/M/1/2* queues when $\lambda_1$ is either very small or very large. 
	Another interesting property obtained from  these simulations is that the average AoI of the
	M/M/1/3* queue is not monotone in $\lambda_1$. This phenomenon is due to the FCFS discipline and the presence of the buffer of size 2 and can be interpreted as follows. For small arrival rates, all the packets will be delivered directly without staying too much time in the buffer and the system will behave like an M/M/1/1 queue. Therefore, when the arrival rate increases and on average there is only one packet (or less) in the server, the AoI will keep decreasing since more fresh packets improves the AoI. We can observe in this figure  that the minimum AoI is achieved when $\lambda_1=\mu$ (which means that on average we have one packet in the server as explained above). Then,  when the arrival rate keeps increasing, there will be always packets in the buffer (in addition to the packet in the server) and the arrived packets will be delayed, which will increase the AoI. When the arrival rate grows very large, the packet in the second place in the buffer will be constantly replaced by the new arrived packet. However, there  will be always a delay due the fact that the packet in the first place of the buffer should wait until the packet in the server is delivered. In other words, the AoI will converge to an asymptotic value. However, this asymptotic value is greater than the AoI when $\lambda=\mu$ since in that case there is on average one packet in the system (i.e. the packet is directly served by the server) and hence the packets are not delayed by the buffer. \textcolor{black}{In addition to these results, we provide in Appendix B some results  obtained by simulations for single M/M/1/3* queue. As expected The results assess the accuracy of the SHS method used to evaluate the average AoI.} 

In Figure~\ref{fig:comparison_mm12_theta0_alpha10}, we study the average AoI
for a system with multiple sources and without losses. We see that the AoI of 
the two parallel M/M/1/2* queues coincides with that of the M/M/1/3* queue with half traffic rate 
when $\lambda$ is small, whereas it coincides with that of  the M/M/1/3* queue with double service
rate when $\lambda$ is large. It is worth mentioning that we consider here that the arrival rate of the rest of the sources is $\sum_{k>1}\lambda_k=10$, which implies that there will be always packets in the server and in the buffer and the system cannot behave as an M/M/1/1 by changing $\lambda_1$ of the first source. This explains why the average AoI decreases with $\lambda_1$ until reaching a limiting value and the average AoI does not have the same shape as in Figure~\ref{fig:comparison_mm12_theta0_alpha0}.

We also study the average AoI with a single source and losses in 
Figure~\ref{fig:comparison_mm12_theta10_alpha0}. For this case, the AoI of 
the system with two parallel M/M/1/2* queues and of an M/M/1/3* queue with double service rate 
coincide when $\lambda_1$ is either small or large. In this case, we can see that the average AoI decreases with the arrival rate $\lambda_1$. This can be explained by the fact that, since the packets can get lost, it is better from AoI perspective to have more arrived packets (even if these packets are delayed in the buffer).

Finally, in Figure~\ref{fig:comparison_mm12_theta10_alpha10}, we show the average age of 
information for different values of $\lambda_1$ when there are multiple sources and losses.
This illustration presents that, depending on the value of $\lambda_1$, the AoI 
approaches that of an M/M/1/3* with half arrival rate and half loss rate or that of 
an M/M/1/3* with double service rate, as in Figure~\ref{fig:comparison_mm12_theta10_alpha10}.

The main conclusion of these illustrations is that, from an AoI perspective,  the M/M/1/3* 
queue with double service rate is the optimal one 
among the systems under consideration. Besides, we characterize the instances where the 
AoI of the two parallel queues coincides with the optimal AoI.

\begin{figure}[t!]
	\centering
	\includegraphics[width=\columnwidth,clip=true,trim=10pt 230pt 0pt 240pt]{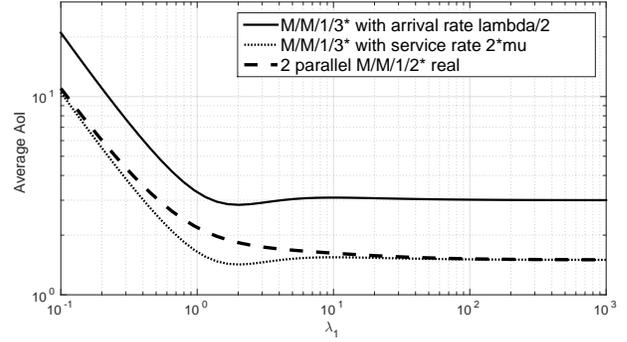}
	\caption{Average AoI comparison when $\lambda_1$ varies from $0.1$ to $10^{3}$ with a single source and without losses ($\lambda_k=0$ for all $k>1$ and $\theta=0$). $\mu=1$.
	}
	\label{fig:comparison_mm12_theta0_alpha0}
\end{figure}

%

\begin{figure}[t!]
	\centering
	\includegraphics[width=\columnwidth,clip=true,trim=10pt 230pt 0pt 240pt]{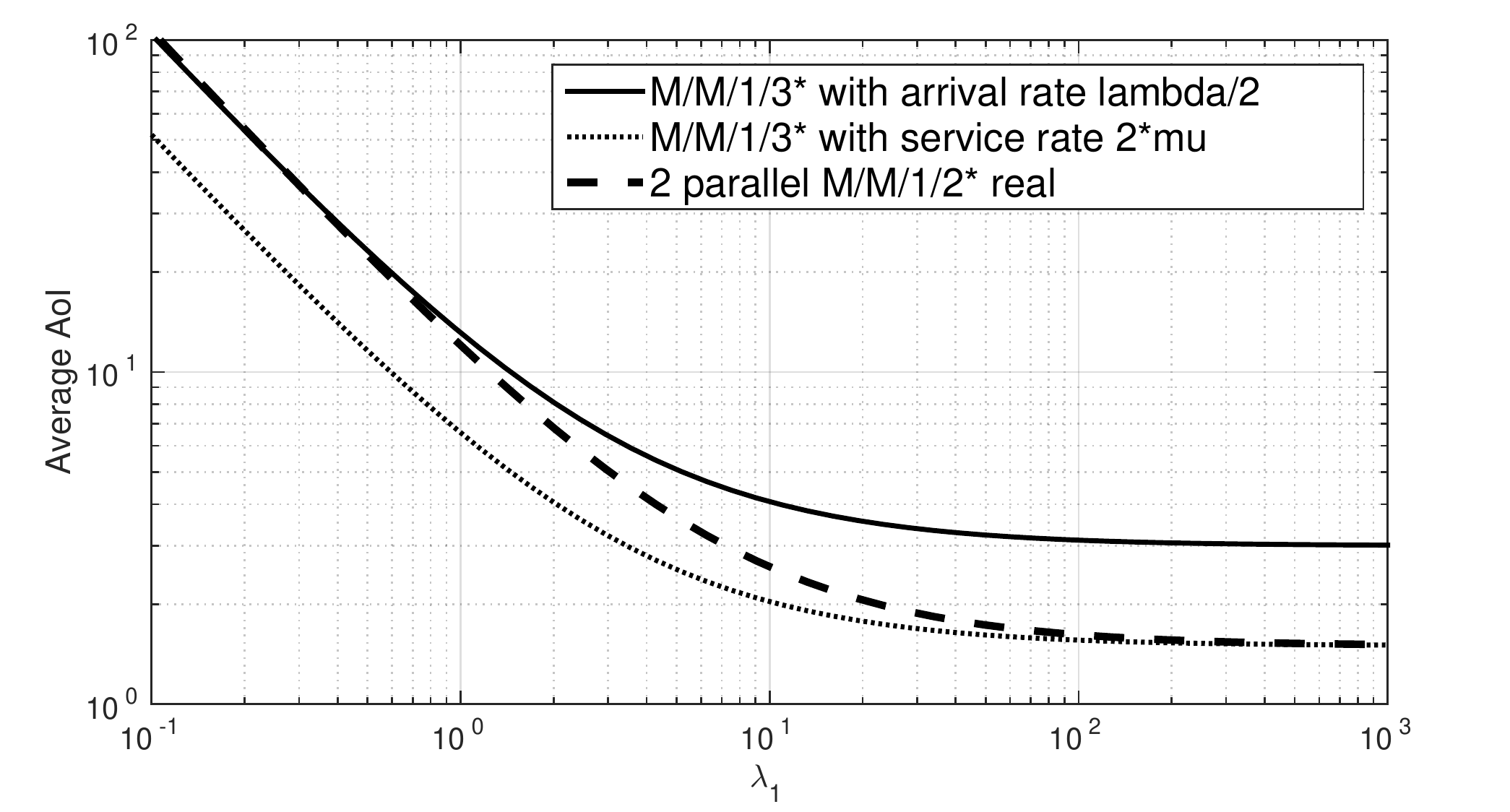}
	\caption{
			Average AoI comparison when $\lambda_1$ varies from $0.1$ to $10^{3}$ 
			with multiple sources and without losses  ($\sum_{k>1}\lambda_k=10$ and $\theta=0$). $\mu=1$.}
	\label{fig:comparison_mm12_theta0_alpha10}
\end{figure}

\begin{figure}[t!]
	\centering
	\includegraphics[width=\columnwidth,clip=true,trim=10pt 230pt 0pt 240pt]{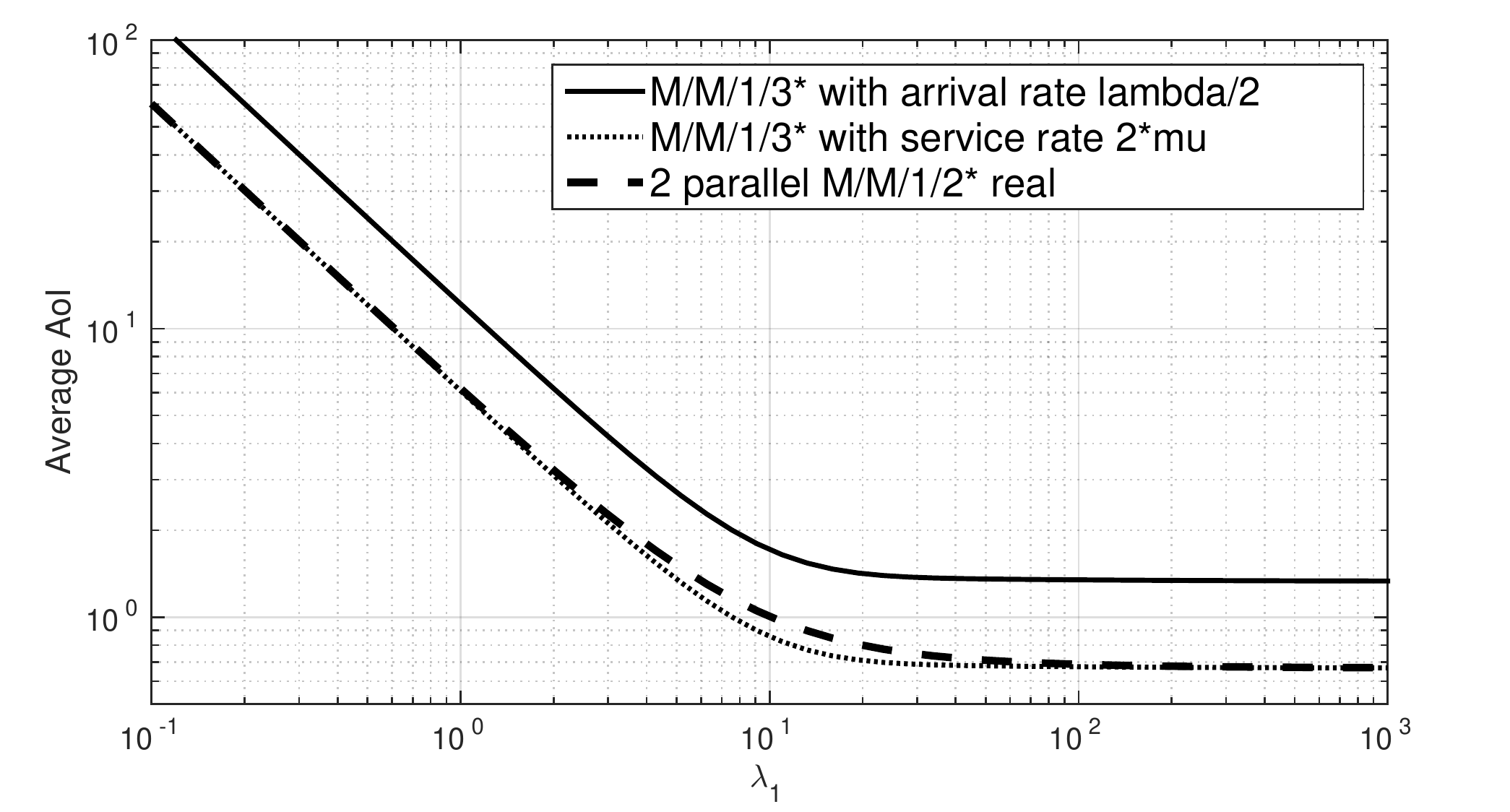}
	\caption{
			Average AoI comparison when $\lambda_1$ varies from $0.1$ to $10^{3}$ 
			with a single source and losses  ($\lambda_k=0$ for all $k>1$ and $\theta=10$). $\mu=1$.}
	\label{fig:comparison_mm12_theta10_alpha0}
\end{figure}

\begin{figure}[t!]
	\centering
	\includegraphics[width=\columnwidth,clip=true,trim=10pt 230pt 0pt 240pt]{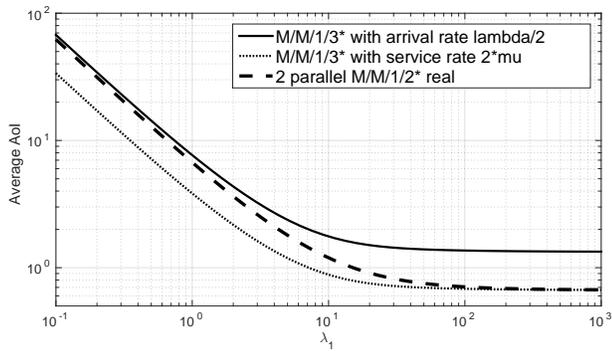}
	\caption{
			Average AoI comparison when $\lambda_1$ varies from $0.1$ to $10^{3}$ 
			with multiple sources and losses  ($\sum_{k>1}\lambda_k=10$ and $\theta=10$). $\mu=1$.}
	\label{fig:comparison_mm12_theta10_alpha10}
\end{figure}

\subsection{Parallel Queues With Buffer Size $N>1$}

We now focus on the study of the average AoI in a system with $K$ parallel queues with buffer size $N>1$. We 
notice that using the SHS method leads to the analysis
of a Markov Chain with a number of states equal to $K\cdot(N+2)$. 
This implies that the number of SHS transitions increases at a very high rate with the number of queues and with the buffer size. 
As a result, according to \eqref{eq:shs-thm}, the number of equations to be solved so as to obtain the AoI 
suffers from the curse of dimensionality. Thus, providing an analytical expression 
of the AoI of source $1$ in a system with an arbitrary routing system (with an arbitrary number of queues and an 
arbitrary buffer size) seems to be intractable using the considered method. However, as we will see in the next section, it is possible 
to provide an upper-bound on the AoI.

\section{Upper-bound on the Average AoI for an Arbitrary Routing System}
\label{sec:bound-aoi}

We study the average AoI of a system with $K>2$ parallel queues with $N>1$ buffer size. In this section, we provide an upper bound on the age 
of information using the SHS method in a system with a single and multiple sources.

	We now explain the system we study here. We consider a system with $n$ sources where, for all i, the updates of source $i$ and
	of the rest of the sources arrive to the system with rate $\lambda_i$. We denote by $p_{ij}$ the probability that a job of source
	$i$ is routed to server $j$. Hence, $\lambda_{i}=\sum_{j=1}^K\lambda_ip_{ij}$. Besides, the total incoming traffic to the 
	system is denoted by $\lambda$, i.e., $\lambda=\sum_{i=1}^n \lambda_i$. We assume that
	the service rate  in queue $j$ is exponentially distributed with rate $\mu_j$.
	In the following result, we provide an upper bound of the average AoI of the 
	system under study here. \textcolor{black}{Without loss of generality,  the result is provided for source $1$, however one can see that the result can be obtained for any source $i$. The proof of this result is reported in Appendix~\ref{proof:thm:upperbound}.}
	\begin{theorem}
		For the aforementioned system, the average AoI of source $1$ is
		upper bounded by
		\begin{equation}
		\frac{1}{\sum_{j=1}^K\mu_j}
		\left(1+KN+\sum_{j=1}^K\frac{\sum_{k>1}\lambda_k p_{kj}+\mu_j}{\lambda_{j}p_{1j}}\right).
		\label{eq:thm:upperbound}
		\end{equation}
		\label{thm:upperbound}
	\end{theorem}

\textcolor{black}{
 It is also important to note that to obtain the above result we consider that when an update completes the service, we create a 
\emph{fake} update to keep the system full of packets. This is the reason why, unlike in the previous section, 
we are not able to study the influence of the packet losses  on the upper bound of the average AoI 
we provide in Theorem~\ref{thm:upperbound}. In fact, let us  consider that $N=0$ as an example. In this case, when a packet in service is lost, 
we  put a false/fake update with the same age of the lost one in service. Therefore, this fake update modifies the age of the monitor
when it is served and delivered to the monitor. As a result, the fake update will modify the age at the monitor and the system will behave like a system with no packet losses (but with a  different service rate). This explains why the above result cannot capture the impact of packet losses. 
}

\subsection{Tightness of the Upper bound}
	We now aim to explore if the upper bound on the average AoI is tight for the 
	systems we have studied in Section~\ref{sec:real-aoi}. We consider $\mu=1$ and, when we
	analyze the AoI for multiple sources, we fix the arrival rate of the rest of the sources 
	to $10$, that is, $\sum_{k>1}\lambda_k=10$.

	We first study in Figures~\ref{fig:tightness_2_mm12_single}-\ref{fig:tightness_2_mm12_multiple},
	a system formed by 2 parallel queues with equal arrival rate and service rate, i.e., 
	$K=2,$ $N=1$, $p_{11}=p_{21}=p_{12}=p_{22}=1/2$
	and $\mu_1=\mu_2=\mu$. For this case, we get from Theorem~\ref{thm:upperbound} that the upper bound for source $1$ is
	$$
	\frac{1}{2\mu}\left(3+\frac{\sum_{k>1}\lambda_k+2\mu}{\lambda_1}\right).
	$$

\begin{figure}[t!]
	\centering
	\includegraphics[width=\columnwidth,clip=true,trim=40pt 240pt 56pt 270pt]{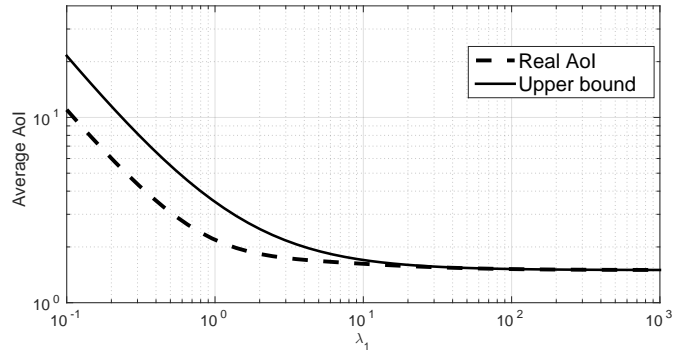}
	\caption{Upper bound and real average AoI comparison when $\lambda$ varies from $0.1$ to $10^{3}$ for two parallel M/M/1/2* queues with a single source ($\lambda_2=0$).
	}
	\label{fig:tightness_2_mm12_single}
\end{figure}

\begin{figure}[t!]
	\centering
	\includegraphics[width=\columnwidth,clip=true,trim=40pt 240pt 56pt 270pt]{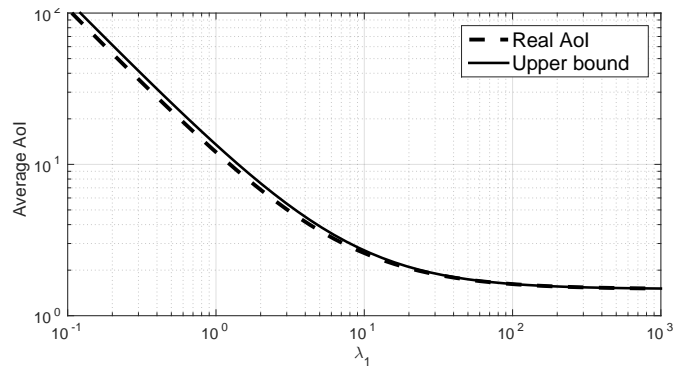}
	\caption{Upper bound and real average AoI comparison when $\lambda$ varies from $0.1$ to $10^{3}$ for two parallel M/M/1/2* queues with multiple sources ($\lambda_2=10$). $\mu=1$.
	}
	\label{fig:tightness_2_mm12_multiple}
\end{figure}

As we observe in Figure~\ref{fig:tightness_2_mm12_single}, the upper bound is tight when 
the arrival rate of source $1$ is large enough, whereas in Figure~\ref{fig:tightness_2_mm12_multiple}, we show that it is always very close to the real age.

	We now focus on the influence of $\mu$ on the tightness of the upper bound of the average AoI in a system with two parallel M/M/1/2* 
	queues with a single source. First, we consider $\mu=0.1$ in Figure~\ref{fig:tightness_2_mm12_single_mu01} and we show that,
	when $\lambda_1$ is larger than 2, the upper bound is very tight. Then, we consider $\mu=10$ in Figure~\ref{fig:tightness_2_mm12_single_mu10} and we show that, when $\lambda_1$ is larger than 10, the upper bound is accurate. Finally,
	we consider $\mu=100$ in Figure~\ref{fig:tightness_2_mm12_single_mu100} and, as it can be observed in this illustration, the upper
	bound is very close to the real age when $\lambda_1$ is larger than 1000.

\begin{figure}[t!]
	\centering
	\includegraphics[width=\columnwidth,clip=true,trim=40pt 240pt 56pt 270pt]{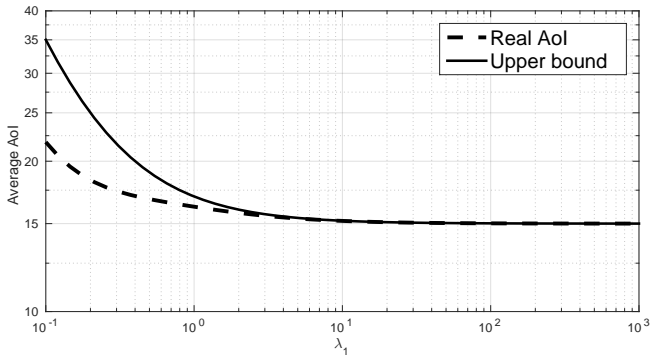}
	\caption{Upper bound and real average AoI comparison when $\lambda$ varies from $0.1$ to $10^{3}$ for two parallel M/M/1/2* queues with a single source ($\lambda_2=0$). $\mu=0.1$.
	}
	\label{fig:tightness_2_mm12_single_mu01}
\end{figure}

\begin{figure}[t!]
	\centering
	\includegraphics[width=\columnwidth,clip=true,trim=40pt 240pt 56pt 270pt]{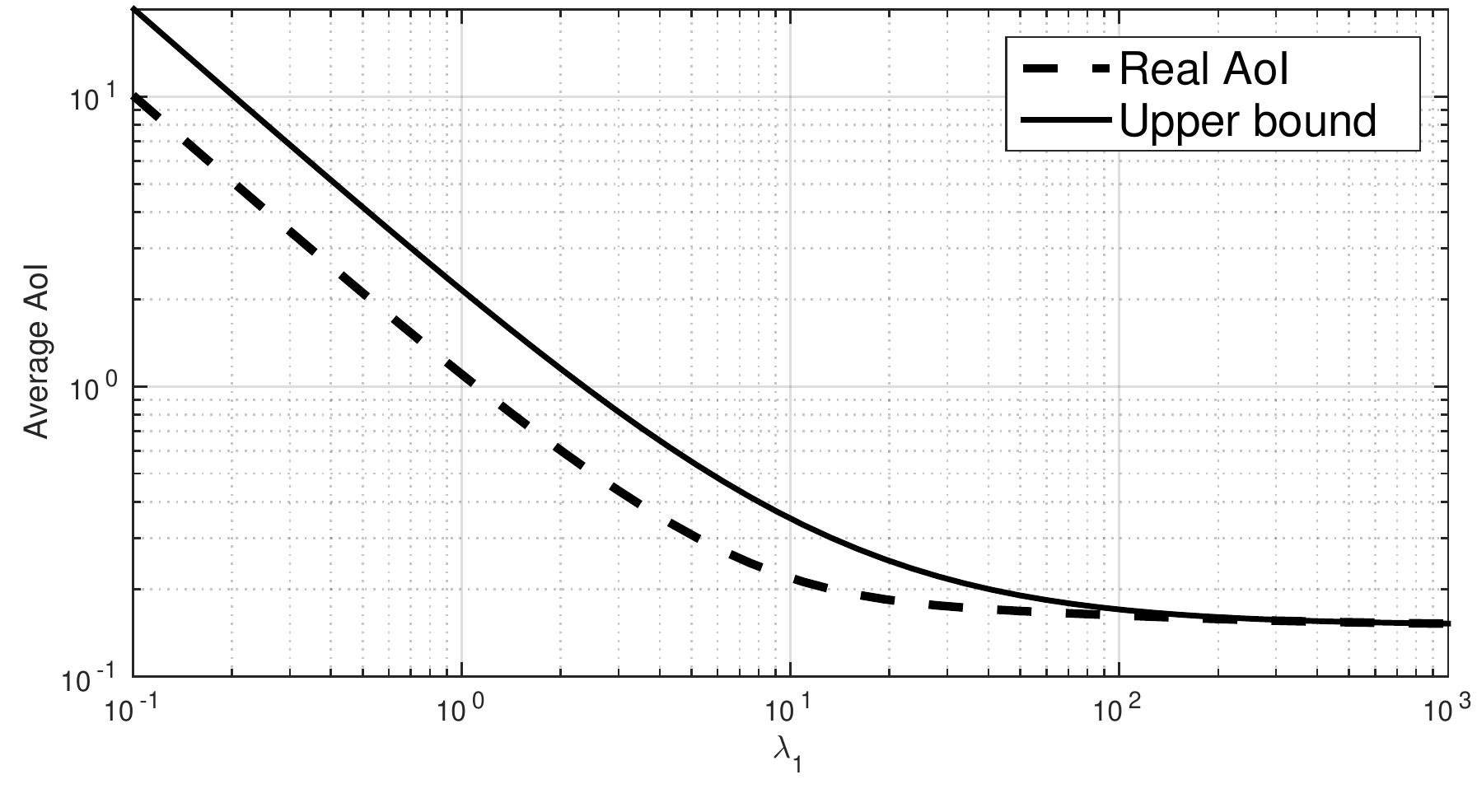}
	\caption{Upper bound and real average AoI comparison when $\lambda$ varies from $0.1$ to $10^{3}$ for two parallel M/M/1/2* queues with a single source ($\lambda_2=0$). $\mu=10$.
	}
	\label{fig:tightness_2_mm12_single_mu10}
\end{figure}

\begin{figure}[t!]
	\centering
	\includegraphics[width=\columnwidth,clip=true,trim=40pt 240pt 56pt 270pt]{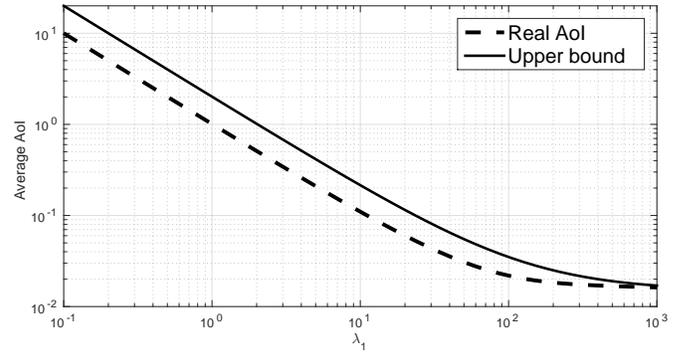}
	\caption{Upper bound and real average AoI comparison when $\lambda$ varies from $0.1$ to $10^{3}$ for two parallel M/M/1/2* queues with a single source ($\lambda_2=0$). $\mu=100$.
	}
	\label{fig:tightness_2_mm12_single_mu100}
\end{figure}

	We also investigate
	a system formed by one M/M/1/3* queue with arrival rate $\lambda$, service rate $\mu$ and multiple sources
	in Figure~\ref{fig:tightness_mm13_multiple}. 
	For this case, we have that $K=1$ and $N=2$ and, therefore, from Theorem~\ref{thm:upperbound}, the average AoI of source $1$ is given by
	$$
	\frac{1}{\mu}\left(3+\frac{\sum_{k>1}\lambda_k+\mu}{\lambda_1}\right).
	$$
	As we see in Figure~\ref{fig:tightness_mm13_multiple}, we show that it is always very close to the real age for any value of $\lambda_1$.
	Thus, this plot confirms that the upper bound on the average AoI we provide in this paper is very tight when the arrival rate is large.

%

\begin{figure}[t!]
	\centering
	\includegraphics[width=\columnwidth,clip=true,trim=40pt 240pt 56pt 260pt]{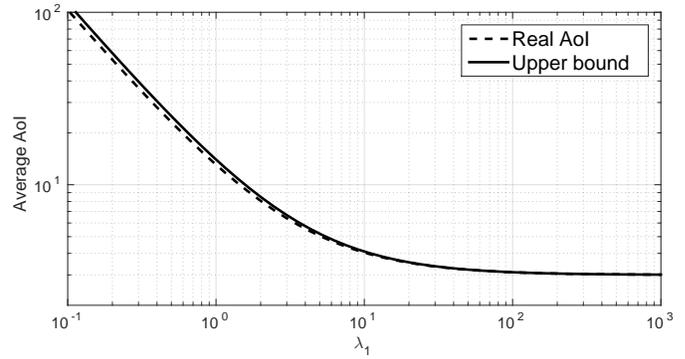}
	\caption{Upper bound and real average AoI comparison when $\lambda$ varies from $0.1$ to $10^{3}$ for one M/M/1/3* queue with multiple source ($\lambda_2=10$). $\mu=1$.
	}
	\label{fig:tightness_mm13_multiple}
\end{figure}

\subsection{AoI Comparison with a Single M/M/1/1 queue}

We now consider a system with a single source which is formed by $K$ homogeneous 
queues without buffer, i.e., $\mu_j=\mu$ for all $j=1,\dots,K$. According to the result of Theorem~\ref{thm:upperbound}, the average AoI is upper bounded by
\begin{equation}
\frac{1}{K\mu}\left(1+\mu\sum_{j=1}^K\frac{1}{\lambda_j}\right).
\label{eq:age-comp-single}
\end{equation}

We now aim to compare the above expression with the average AoI of  a single M/M/1/1 queue with preemption of jobs in service, arrival rate $\lambda/K$ and 
service rate $\mu$, which according to Theorem 2(a) in \cite{YK19}  is given by
$$
\frac{K}{\lambda}+\frac{1}{\mu}.
$$

In the following result, we compare the above expressions.

\begin{proposition}
	Let $K>1$. Then,
	$$
	\frac{K}{\lambda}+\frac{1}{\mu}>\frac{1}{K\mu}\left(1+\mu\sum_{j=1}^K\frac{1}{\lambda_j}\right).
	$$
\end{proposition}

\begin{proof}
	First, we note that, when $\lambda_j=\lambda/K$ for all $j$, we have that
	$$
	\frac{1}{K\mu}\left(1+\mu\sum_{j=1}^K\frac{1}{\lambda_j}\right)=
	\frac{1}{K\mu}\left(1+\frac{\mu K^2}{\lambda}\right)=\frac{1}{K\mu}+\frac{K}{\lambda}.
	$$
	
	Therefore, we aim to show that
	$$
	\frac{1}{K\mu}+\frac{K}{\lambda}<\frac{1}{\mu}+\frac{K}{\lambda} \iff K>1
	$$
	
	And the desired result follows since the last expression is always true.
\end{proof}

An interesting result is derived from the above proposition. Indeed, when we consider a 
system formed by K parallel queues and each of them receives the same arrival rate, 
since the expression \eqref{eq:age-comp-single} provides an upper bound on the average AoI, this result implies that the average AoI 
of a single M/M/1/1 queue with arrival rate $\lambda/K$ is larger than that of the considered system.

\subsection{AoI Comparison with \cite{Y18}}

In Theorem 2 in \cite{Y18}, the author provides the following expression of the average AoI of a system with homogeneous parallel queues where the incoming jobs are always sent to the server with the oldest job:
\begin{equation}
\frac{1}{\mu}\left(\frac{1}{K}\prod_{i=1}^{K-1}\frac{\rho}{i+\rho}+\frac{1}{\rho}+
\frac{1}{\rho}\sum_{l=1}^{K-1}\prod_{i=1}^l\frac{\rho}{i+\rho}\right).
\label{eq:yates-hom}
\end{equation}

We now notice that, in our model, the knowledge of the queue with the oldest job is not 
considered. Therefore, one might expect that the average AoI is always 
smaller in the model in \cite{Y18}. In the following result, we consider the regime where 
$\lambda$ tends to infinity and we compare both models.

\begin{proposition}
	When $\lambda\to\infty$, we have that \eqref{eq:yates-hom} and \eqref{eq:thm:upperbound}
	tend to $\tfrac{1}{K\mu}.$
\end{proposition}

\begin{proof}
	The proof is straightforward from \eqref{eq:yates-hom} and \eqref{eq:thm:upperbound}.
\end{proof}

From this result, we conclude that, when $\lambda\to\infty$, the improvement of the average
AoI caused by the knowledge of the state of the queues is negligible. 

\subsection{Optimization of the Upper Bound}
In this section, we provide a framework to minimize the upper bound of the average age of each source $i$. The objective is to find the routing probabilities $p_{ij}$ in such that the age upper bound for each source $i$ is minimized. The problem can be formulated as a game framework. More formally, the problem can be formulated as
\begin{equation}
\min_{\mathbf{p}_{i}}\frac{1}{\sum_{j=1}^K\mu_j}
\left(1+KN+\sum_{j=1}^K\frac{\sum_{l\neq i}^n\lambda_{l}p_{lj}+\mu_j}{\lambda_{i}p_{ij}}\right), \ \forall i \label{game:equ}
\end{equation} 
\begin{equation}
\text{s.t.} \ \sum_{j=1}^K p_{ij}=1 \ \forall i
\end{equation}
where $\mathbf{p}_i=[p_{i1},...,p_{iK}].$ In the sequel, we will characterizes the   Nash Equilibrium (NE) of the above problem and provide a solution that achieves the NE. Before defining NE, we first introduce the so-called best response set valued function ($BR_i$) for each source (or player) $i$, which is given as follows
\[\text{($BR_i$)}: \ \text{Given} \ \mathbf{p}_{-i} \overset{\Delta}{=} (\mathbf{p}_1,...,\mathbf{p}_{i-1},\mathbf{p}_{i+1},...\mathbf{p}_N)\]
\[\mathbf{p}_i \in \arg \max_{\mathbf{p}'_i} \ U_i(\mathbf{p}'_i,\mathbf{p}_{-i})\]
where $\mathbf{p}'_i$ satisfies $\sum_{j=1}^K p'_{ij}=1$ and  $U_i(\mathbf{p}'_i,\mathbf{p}_{-i})=\frac{1}{\sum_{j=1}^K\mu_j}
\left(1+KN+\sum_{j=1}^K\frac{\sum_{l\neq i}^N\lambda_{l}p_{lj}+\mu_j}{\lambda_{i}p_{ij}}\right)$. In order to show explicitly the dependence of $BR_i$ in $\mathbf{p}_{-i}$, we will use the notation $BR_i(\mathbf{p}_{-i})$ to represent the best response set valued function of source $i$. In other words, the best response consists of optimizing the utility of each source with respect only to its own action vector (i.e. routing probability $\mathbf{p}_i$).\\
Furthermore, we can also define the  sources’ joint best-response function as 
\[BR(\mathbf{p})=\left(BR_1(\mathbf{p}_{-1}),...BR_N(\mathbf{p}_{-N})\right)\]
We now provide the definition of NE and the relation with the sources' best response.
\begin{definition}
	A strategy profile $\mathbf{p}=(\mathbf{p}_1,...,\mathbf{p}_N)$ is a pure Nash equilibrium iff $\forall i$,\\
	\[\forall \ \mathbf{p}'_i, \ U_i(\mathbf{p}'_i,\mathbf{p}_{-i}) \geq U_i(\mathbf{p}_i,\mathbf{p}_{-i})\]  
\end{definition}
\begin{definition}
	A strategy profile is a Nash Equilibrium iff 
	\[\mathbf{p} \in BR(\mathbf{p})\]
\end{definition}
In words, a NE is a fixed point of the BR dynamic. In the sequel, we will therefore show that a fixed point of BR exists and it is unique. \\
It is straightforward to see that $U_i(\mathbf{p}'_i,\mathbf{p}_{-i})$ is convex with respect to $\mathbf{p}_i$. The best response problem for each source $i$ can be solved easily using the standard Lagrangian technique. \\
The Lagrangian for each source can be written as follows, 
\[L_i(\mathbf{p}_i,\delta)=U_i(\mathbf{p}_i,\mathbf{p}_{-i})+\delta (\sum_{j=1}^K p_{ij}-1) \]
The optimal solution of the above optimization problem for each source $i$, i.e. ($\mathbf{p}^*_i,\delta^*$), can then be obtained by KKT and complementarity conditions, i.e. by using $\frac{\partial L_i}{\partial p_{ij}}=0$ $\forall j$ and $\delta^* (\sum_{j=1}^K p^*_{ij}-1)=0$. After some algebraic manipulations, this leads to the following expression of $\mathbf{p}^*_i$ 
\[p_{ij}^*=\frac{\omega_{ij}}{\sum_{j'=1}^K \omega_{ij'}}\]
where $\omega_{ij}=\sqrt{\sum_{l\neq i}^N\lambda_{l}p_{lj}+\mu_j}$.\\
Consequently, a NE is simply the solution of the following system of equations
\begin{equation}
p_{ij}^*=\frac{\sqrt{\sum_{l\neq i}^N\lambda_{l}p^*_{lj}+\mu_j}}{\sum_{j'=1}^K \sqrt{\sum_{l\neq i}^N\lambda_{l}p^*_{lj'}+\mu_{j'}}}, \ \forall i,j \label{BR:equ}
\end{equation}


We will show numerically later on that the aforementioned system of equations has a solution. Before that, we will provide an analysis in the case of large number of sources, by using mean field analysis, and develop a simple iterative algorithm allowing each source to find its probabilistic routing vector $\mathbf{p}_i=[p_{i1},...,p_{iK}]^T$. 

\subsubsection{Mean Field Analysis}
We provide here an analysis in the case of large number of sources and  provide a distributed iterative algorithm that converges to the solution of the system of equations  in (\ref{BR:equ}). In order to use mean field tools, indistinguishable sources should be considered. We therefore consider that $\lambda_l=\bar{\lambda}$. We first define the following mean field term for each queue $j$: 
\[m^n_{ij}=\frac{1}{n}\sum_{l\neq i}^{n}p_{lj}\]
One can notice that $m^n_{ij} \rightarrow m_j$ when $n \rightarrow \infty$ $\forall$ $i$, where $m_j=\lim_{n\to\infty}\frac{1}{n}\sum_{i=1}^{n}p_{ij}$.\\
The problem in (\ref{game:equ}) can be written as 
\[\min_{\mathbf{p}_{i}}\frac{1}{\sum_{j=1}^K \frac{\mu_j}{n}}
\left(\frac{1+KN}{n}+\sum_{j=1}^K\frac{\lambda m^n_{ij}+\frac{\mu_j}{n}}{\lambda p_{ij}}\right), \ \forall i \]
One can notice also that when $n \rightarrow \infty$, $\mu_j/n \rightarrow \bar{\mu}_j$, where $\bar{\mu}_j$ represents the scaled asymptotic service rate per source. In fact, when the number of sources tends to infinity, the service rate of each queue must be high enough to serve all sources. \\
When $n$ is very large, the system of equations in (\ref{BR:equ}) tends to:
\begin{equation}
p_{ij}^*=\frac{\sqrt{m_j+\frac{\bar{\mu}_j}{\bar{\lambda}}}}{\sum_{j'=1}^K \sqrt{m_{j'}+\frac{\bar{\mu}_{j'}}{\bar{\lambda}}}}, \ \forall i,j \label{BR1:equ}
\end{equation}
By summing the previous equations over all source indexes $i$ and taking the limit when $n \rightarrow \infty$, we get
\begin{equation}
m_j=\frac{\sqrt{m_j+\frac{\bar{\mu}_j}{\bar{\lambda}}}}{\sum_{j'=1}^K \sqrt{m_{j'}+\frac{\bar{\mu}_{j'}}{\bar{\lambda}}}}, \ \forall j \label{MF:equ}
\end{equation}
By making the variable change $y_j=\sqrt{m_j+\frac{\bar{\mu}_j}{\bar{\lambda}}}$, the previous equations can be written as: 
\begin{equation}
y_j=\frac{1}{\sum_{j'=1}^K y_{j'}}+\frac{\bar{\mu}_j}{\bar{\lambda} y_j}, \ \forall j \label{MF1:equ}
\end{equation}
\begin{proposition}
	If the system of equations in (\ref{MF1:equ}) has a solution, then this solution is unique. 
\end{proposition}
\begin{proof}
	Let's assume that the system of equations has two different solutions $\mathbf{y}^1=[y^{(1)}_1,...,y^{(1)}_K]^T$ and $\mathbf{y}^2=[y^{(2)}_1,...,y^{(2)}_n]^T$. Without loss of generality, we can consider that there exists $\nu>1$ such that $\mathbf{y}^1 \leq \nu \mathbf{y}^2$ and $\exists$ at least one $j$ for which $y^{(1)}_j <\nu y^{(2)}_j$ and $\exists$ at least one queue $k$ for which $y^{(1)}_k =\nu y^{(2)}_k$. One can see easily that for any possible vectors  $\mathbf{y}^1=[y^{(1)}_1,...,y^{(1)}_K]^T$ and $\mathbf{y}^2=[y^{(2)}_1,...,y^{(2)}_n]^T$  obtaining $\nu$ to satisfy the above statement is straightforward. Recall that for queue $k$ 
	\[y^{(2)}_k=\frac{1}{\sum_{j'=1}^K y^{(2)}_{j'}}+\frac{\bar{\mu}_k}{\bar{\lambda} y^{(2)}_k} \]
	By using $\mathbf{y}^1 \leq \nu \mathbf{y}^2$ and  $y^{(1)}_j <\nu y^{(2)}_j$  for at least one queue $j$, we get  
	\[y^{(2)}_k < \frac{\nu}{\sum_{j'=1}^K y^{(1)}_{j'}}+\frac{\nu \bar{\mu}_k}{\bar{\lambda} y^{(1)}_k}=\nu y^{(1)}_k \] 
	Therefore, we obtain  $y^{(2)}_k < \nu  y^{(1)}_k$ which contradicts the fact that $y^{(2)}_k = \nu  y^{(1)}_k$. Consequently, it is not possible to have two different solutions to the system of equations in (\ref{MF1:equ}).
\end{proof}
In order to solve the system of nonlinear equations in (\ref{MF1:equ}), we provide an efficient iterative learning algorithm and prove its convergence to the solution of (\ref{MF1:equ}). The proposed algorithm has a reduced complexity and can be implemented in a distributed manner. We first consider the following iterative algorithm:  
\begin{equation}
y_j(t+1)=(1-\alpha)y_j(t)+\alpha \frac{1}{\sum_{j'=1}^K y_{j'}(t)}+\alpha \frac{\bar{\mu}_j}{\bar{\lambda} y_j(t)}, \ \forall j \label{Mann-iteration}
\end{equation}
where $\alpha$ is a sufficiently small step size. 
This algorithm is a simple class of Ishikawa algorithm  (using Mann-like iteration but with constant step size $\alpha$; see \cite{Ishikawa} for details). \\
By definition of $\mathbf{y}$, one can notice that $\forall j$ $y_j \in \mathcal{S}_j= \biggl [\sqrt{\frac{\bar{\mu}_j}{\bar{\lambda}}},\sqrt{1+\frac{\bar{\mu}_j}{\bar{\lambda}}} \biggr ]$. Let $\mathcal{S}=\mathcal{S}_1\times ...\times \mathcal{S}_K$ be the set of feasible values of $\mathbf{y}$. In order to ensure that the solution obtained by the iterative algorithm lies in the feasible set $\mathcal{S}$, we consider  the following projection $\Pi_\mathcal{S}$: $\hat{\mathbf{y}}= \Pi_\mathcal{S} (\mathbf{y})$ defined as, \\ $\forall j$, if $y_j < \sqrt{\frac{\bar{\mu}_j}{\bar{\lambda}}}$ then $\hat{y}_j= \max{\{y_j,\sqrt{\frac{\bar{\mu}_j}{\bar{\lambda}}}\}}$; if $y_j > \sqrt{1+\frac{\bar{\mu}_j}{\bar{\lambda}}}$ then $\hat{y}_j=\min{\{y_j,\sqrt{1+\frac{\bar{\mu}_j}{\bar{\lambda}}}\}}$; and $\hat{y}_j=y_j$ otherwise. Using the aforementioned projection, the iterative algorithm becomes
\begin{equation}
\hat{y}_j(t+1)=\Pi_\mathcal{S} \biggl ( (1-\alpha)\hat{y}_j(t)+\alpha \frac{1}{\sum_{j'=1}^K \hat{y}_{j'}(t)}+\alpha \frac{\bar{\mu}_j}{\bar{\lambda} \hat{y}_j(t)} \biggr), \ \forall j \label{Mann-iteration1}
\end{equation}
The following result proves that the algorithm above converges to the solution of the system of equations in (\ref{MF1:equ}), whenever it exists. Notice that once the algorithm converges, one can obtain $\mathbf{m}$ from $\mathbf{y}$ from the relation $y_j=\sqrt{m_j+\frac{\bar{\mu}_j}{\bar{\lambda}}}$ $\forall j$. The routing probabilities for each source $i$, i.e. $\mathbf{p}_i$, can then be obtained from (\ref{BR1:equ}). Finally, one can see that since the aforementioned algorithm depends only on the average arrival and service rates, without requiring any information about the instantaneous status of the network, it can be implemented separately by each source.

\begin{proposition}
	The algorithm in (\ref{Mann-iteration1}) converges  to the solution of the system of equations in (\ref{MF1:equ}). 
\end{proposition}
\begin{proof}	
 We denote by $f_j(\mathbf{y}(t))= \frac{1}{\sum_{j'=1}^K y_{j'}(t)}+\frac{\bar{\mu}_j}{\bar{\lambda} y_j(t)}$. One can see that $f_j(\mathbf{y}(t))$ can be written as $f_j(\mathbf{y}(t))=g(\mathbf{y}(t))+g_j(y_j(t))$ where $\tilde{g}(\mathbf{y}(t))=\frac{1}{\sum_{j'=1}^K y_{j'}(t)}$ and $g_j(y_j(t))=\frac{\bar{\mu}_j}{\bar{\lambda} y_j(t)}$. We also denote by $\mathbf{y}^*=[y^*_1,...,y^*_K] \in \mathcal{S}$ the solution of (\ref{MF1:equ}). Recall that by using the algorithm in (\ref{Mann-iteration1}), we have $y_j(t+1))=(1-\alpha)\hat{y}_j(t)+\alpha \frac{1}{\sum_{j'=1}^K \hat{y}_{j'}(t)}+\alpha \frac{\bar{\mu}_j}{\bar{\lambda} \hat{y}_j(t)}$ and $\hat{y}_j(t+1)=\Pi_\mathcal{S}(y_j(t+1))$.  We can also write the following 
		\begin{align}
		& y_j(t+1)-y_j^*  = (1-\alpha)(\hat{y}_j(t)-y_j^*)+ \nonumber \\ & \alpha (\tilde{g}(\hat{\mathbf{y}}(t))-\tilde{g}(\mathbf{y}^*))  +\alpha_t (g_j(\hat{y}_j(t))-g_j(y_j^*))
		\end{align}
		By using the mean value theorem, $\exists \ \mathbf{\bar{y}}(t) \in [\hat{\mathbf{y}}(t),\mathbf{y}^*]$ such that 
		\[\tilde{g}(\hat{\mathbf{y}}(t))-\tilde{g}(\mathbf{y}^*)=\nabla \tilde{g} |_{\mathbf{\bar{y}}(t)}^T \left(\hat{\mathbf{y}}(t)-\mathbf{y}^*\right)\]
		Similarly,  $\exists \ \mathbf{\tilde{y}}(t)=[\tilde{y}_1(t),...,\tilde{y}_K(t)]^T \in [\hat{\mathbf{y}}(t),\mathbf{y}^*]$ such that 
		\[g_j(\hat{y}_j(t))-g_j(y_j^*)=\frac{dg_j}{dy_j} |_{\tilde{y}_j(t)}  (\hat{y}_j(t)-y_j^*), \ \forall j \]
		From all the above, we can therefore write $\mathbf{y}(t+1)-\mathbf{y}^*$ as follows:
		\begin{align}
		& \mathbf{y}(t+1)-\mathbf{y}^*  = (1-\alpha)(\hat{\mathbf{y}}(t)-\mathbf{y}^*)+ \alpha \mathbf{J}_t (\hat{\mathbf{y}}(t)-\mathbf{y}^*) 
		\end{align}
where $\mathbf{J}_t$ is a $K\times K$ matrix with diagonal elements $-\left(\frac{1}{\sum_{l=1}^{K}\bar{y}_j(t)} \right)^2-\frac{\bar{\mu}_j}{\bar{\lambda} \tilde{y}^2_j(t)}$ and the non-diagonal elements are $-\left(\frac{1}{\sum_{l=1}^{K}\bar{y}_j(t)} \right)^2$. We can show  easily then that $\forall \mathbf{a} \in \mathbbm{R}^K$ $\mathbf{a}^T \mathbf{J}_t \mathbf{a} \leq 0$, i.e. $\mathbf{J}_t$ is negative semi definite. Then, by using $\|\mathbf{y}(t+1)-\mathbf{y}^*\|^2=(\mathbf{y}(t+1)-\mathbf{y}^*)^T (\mathbf{y}(t+1)-\mathbf{y}^*)$, we get  
	\begin{align}
	& \|\mathbf{y}(t+1)-\mathbf{y}^*\|^2  = (1-\alpha)^2\|\hat{\mathbf{y}}(t)-\mathbf{y}^*\|^2+ \nonumber \\& \alpha^2 \| \mathbf{J}_t (\hat{\mathbf{y}}(t)-\mathbf{y}^*) \|^2+ 2(1-\alpha_t)\alpha (\hat{\mathbf{y}}(t)-\mathbf{y}^*)^T \mathbf{J}_t (\hat{\mathbf{y}}(t)-\mathbf{y}^*) 
	\end{align}	
	Since $\mathbf{J}$ is negative semi definite and $\| \mathbf{J}_t (\hat{\mathbf{y}}(t)-\mathbf{y}^*) \|^2 \leq \| \mathbf{J}_t\|^2 \|\hat{\mathbf{y}}(t)-\mathbf{y}^* \|^2$, we get 
	\begin{align}
	 \|\mathbf{y}(t+1)-\mathbf{y}^*\|^2  \leq \biggl( (1-\alpha)^2+\alpha^2 \| \mathbf{J}_t\|^2 \biggr) \|\hat{\mathbf{y}}(t)-\mathbf{y}^*\|^2
	\end{align}	
	Recall that $\hat{\mathbf{y}} \in \mathcal{S}$  and therefore $\hat{\mathbf{y}}(t) \neq 0$, which implies that $\bar{\mathbf{y}} \neq 0$ and $\tilde{\mathbf{y}} \neq 0$  and hence $\| \mathbf{J}_t\|^2$ is bounded. Therefore, $\exists$ a sufficiently small $\alpha$  such that $\forall t$  $\alpha^2 \| \mathbf{J}_t\|^2 +\alpha^2 -2\alpha < -\epsilon <0$, 	where $0<\epsilon <1$. Consequently, 
	\[	\|\mathbf{y}(t+1)-\mathbf{y}^*\|^2  \leq \biggl( 1-\epsilon \biggr) \|\hat{\mathbf{y}}(t)-\mathbf{y}^*\|^2\]
	Then, by using the inequality 	$\|\hat{\mathbf{y}}(t+1)-\mathbf{y}^*\|^2 \leq 	\|\mathbf{y}(t+1)-\mathbf{y}^*\|^2$ (since $\mathbf{y}^* \in \mathcal{S}$), we get 
	\begin{align}
	\|\hat{\mathbf{y}}(t+1)-\mathbf{y}^*\|^2 & \leq \biggl( 1-\epsilon \biggr) \|\hat{\mathbf{y}}(t)-\mathbf{y}^*\|^2 \nonumber \\ & \vdotswithin{\leq} \nonumber \\ &  \leq \biggl( 1-\epsilon \biggr)^{t+1} \|\hat{\mathbf{y}}(0)-\mathbf{y}^*\|^2 
	\end{align}	
	and $\|\hat{\mathbf{y}}(t)-\mathbf{y}^*\|^2 \rightarrow 0$ when $t\rightarrow \infty$. This concludes the proof. 
\end{proof}
From the Proof above, one can see that if $\alpha$ is taken such that  $\alpha^2 \| \mathbf{J}_t\|^2 +\alpha^2 -2\alpha <0$ then the algorithm in (\ref{Mann-iteration1}) converges to the solution of (\ref{MF1:equ}). Since for each $j$ $\hat{y}_j \geq \sqrt{\frac{\bar{\mu}_j}{\bar{\lambda}}}$ and $\|\mathbf{J}_t\|^2 \leq \sum_{i=1}^K\sum_{j=1}^K J_t(i,j)$, which implies that $\|\mathbf{J}_t\|^2 \leq K^2 \left(\frac{1}{\sum_{l=1}^{K}\sqrt{\frac{\bar{\mu}_j}{\bar{\lambda}}} } \right)^2+K$. Consequently, it is sufficient to take $\alpha <\frac{2}{K^2 \left(\frac{1}{\sum_{l=1}^{K}\sqrt{\frac{\bar{\mu}_j}{\bar{\lambda}}} } \right)^2+K+1}$.

\subsubsection{Numerical Results}
We provide here an example to show numerically that the  system of equations in (\ref{BR:equ}) has a solution. \textcolor{black}{While in the previous subsection, a detailed mathematical analysis is provided to optimize the upper bound in the case of large number of sources, we show here numerical results for finite number of sources. In order to solve (\ref{BR:equ}), we consider the following iterative algorithm:} $p_{ij}(t+1)=(1-\alpha)p_{ij}(t)+\alpha \frac{\sqrt{\sum_{l\neq i}^N\lambda_{l}p_{lj}(t)+\mu_j}}{\sum_{j'=1}^K \sqrt{\sum_{l\neq i}^N\lambda_{l}p_{lj'}(t)+\mu_{j'}}}, \ \forall i,j$. \textcolor{black}{One can see that actually this iterative algorithm is similar to (\ref{Mann-iteration1}), or more precisely (\ref{Mann-iteration1}) can be obtained from the above algorithm by using Mean Field
analysis (using the derivations in the previous subsection). } In Figure~\ref{fig:convergence_fixed_point_optimal_prob}
	we show an illustrative example where we consider a system with 6 sources et 10 servers with the following parameters for the sources $\lambda_1=100$, $\lambda_2=20$, $\lambda_3=50$, $\lambda_4=\lambda_5=10$ and $\lambda_6=1000$ and for the servers $\mu_1=1,$ $\mu_2=2,$ $\mu_3=3$, 
	$\mu_4=5$, $\mu_5=10$, $\mu_6=20$, $\mu_7=50$, $\mu_8=100$, $\mu_9=200$ and $\mu_{10}=1000$. The plots  show that (\ref{BR:equ})  has a solution, 
	that $p_{1,2}$, $p_{1,6}$, $p_{1,7}$, $p_{1,8}$, $p_{1,9}$ and $p_{1,10}$ converge to the solution of (\ref{BR:equ}) and also that in this example the algorithm converges quickly. \textcolor{black}{This shows that the proposed algorithm can converge also for finite number of sources with different arrival rates. Providing a formal proof of convergence would be an interesting topic and is left for future work.}

\begin{figure}[t!]
	\centering
	\includegraphics[width=\columnwidth,clip=true,trim=0pt 240pt 20pt 220pt]{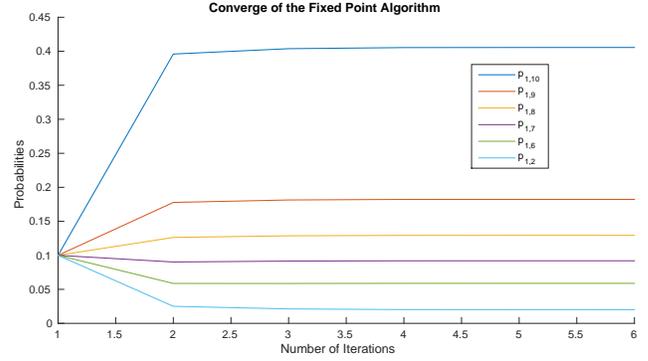}
	\caption{Convergence of the fixed point algorithm defined from \eqref{BR:equ}.}
	\label{fig:convergence_fixed_point_optimal_prob}
\end{figure}

\section{Conclusion}
\label{sec:conclusion}

In this paper, we studied the average AoI of a system of multiple sources and parallel queues using the SHS method. We considered that the queues do not communicate between each other and that the sources send their updates to the queues according to a predefined  probabilistic routing scheme. First, we computed the average AoI for the following systems: i)
two parallel M/M/1/1 queues, ii)  one M/M/1/1 queue with half arrival and  loss rates,  and iii)  one M/M/1/1
queue with double service rate. Then, we computed the average AoI for
two parallel M/M/1/2* queues,  one M/M/1/3* queue with half arrival  and loss rates,  and  one M/M/1/3* queue with double
service time. We conclude that the average AoI of the system composed of parallel queues  is always smaller than that of one 
queue with half arrival  and loss rates, and can be as small as that of one queue with double service rate. We also studied the
average AoI of a system with an arbitrary number of heterogeneous M/M/1/(N+1)* queues and we provided an upper
bound of AoI that is tight when there are multiple sources. We then provided a framework allowing each source to determine its routing decision, by using Game Theory and best response method. In the contest of large number of sources, we simplified the game framework by using Mean Field Games, provided a simple distributed algorithm and proved its convergence to the desired fixed point. 




\bibliographystyle{IEEEtran}
\bibliography{IEEEabrv,routing_aoi}
%

\appendices
\section{Proof of Theorem~\ref{thm:upperbound}}
\label{proof:thm:upperbound}

\begin{table*}[t!]
	\centering
	\begin{tabular}{l l l l l}
		$l$  & $\lambda^{l}$ & x & $x^\prime=x A_l$ & $\bar v_{q_l}A_l$\\\hline
		0 &  $\lambda_{1}p_{1j}$ &$[x_0\ \dots\ x_{j1}\dots\ x_{jN}\ x_{jM}\ \dots  ]$ & $[x_0\ \dots\ x_{j1}\dots\ x_{jN}\ 0\ \dots  ]$ & $[v_0\ \dots\ v_{j1}\dots\ v_{jN}\ 0\ \dots  ]$ \\
		1 &  $\sum_{k>1}\lambda_{k}p_{kj}$ &$[x_0\ \dots\ x_{j1}\dots\ x_{jN}\ x_{jM}\ \dots  ]$ & $[x_0\ \dots\ x_{j1}\dots\ x_{jN}\ x_{jN}\ \dots  ]$ & $[v_0\ \dots\ v_{j1}\dots\ v_{jN}\ v_{jN}\ \dots  ]$ \\
		2 &  $\mu_j $ &$[x_0\ \dots\ x_{j1}\dots\ x_{jN}\ x_{jM}\ \dots  ]$ & $[x_{j1}\ \dots\ x_{j2}\dots\ x_{jM}\ x_{jM}\ \dots  ]$ & $[v_0\ \dots\ v_{j1}\dots\ v_{jM}\ v_{jM}\ \dots  ]$ \\
	\end{tabular}
	\caption{Table of SHS transitions of Figure~\ref{fig:upperbound}.}
	\label{tab:upperbound}
\end{table*}

\textcolor{black}{
We consider a system with $K$ parallel M/M/1/(N+1)* queues. Without loss of generality, we compute the 
average AoI of source 1. The result is of course true for any source $i$. 
The arrival rate to queue $j$ from
source $1$ is $\lambda_1p_{1j}$ and that from the rest of the sources is 
$\sum_{k>1}\lambda_kp_{kj}$ for all $j=1,\dots,K$.}
Let $M=N+1.$ The continuous state is given by a vector of size $1+K*M$
\begin{multline*}
\mathbf x(t)=[x_0(t)\ x_{11}(t)\ \dots\ x_{1M}(t)\ x_{21}(t)\ \ x_{K1}(t)\ \dots\ \\x_{KM}(t)\ ],
\end{multline*}
where $x_0$ represents the current age, $x_{j1}$ the age if the update in service in the queue $j$ is delivered and $x_{jl}$ the age if the update in the position 
$l-1$ of the queue $j$ is delivered. The discrete state is a Markov Chain with a single state. 
We note that, when an event (arrival or departure) occurs in queue $j$, the age of the updates
in the rest of the queues is not modified. This allows us to focus on a queue $j$ to illustrate 
the Markov Chain and the SHS transitions, which are presented respectively in 
Figure~\ref{fig:upperbound} and Table~\ref{tab:upperbound}.

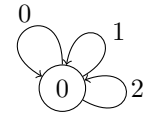
\begin{figure}[t!]
	\centering
	\begin{tikzpicture}[]
	\node[style={circle,draw}] at (0,0) (1) {$0$};
	\node[] at (-0.5,1) (X) {$0$};
	\node[] at (1,0) (X) {$2$};
	\node[] at (0.75,0.75) (X) {$1$};
	\draw[->] (1) to [out=330,in=20,looseness=8] (1) ;
	\draw[->] (1) to [out=30,in=80,looseness=8] (1) ;
	\draw[->] (1) to [out=90,in=150,looseness=8] (1) ;
	\end{tikzpicture}  
	\caption{The SHS Markov Chain for a system formed by $K$ parallel M/M/1/(N+1)* queues with multiple sources.}
	\label{fig:upperbound}
\end{figure}

We now explain each transition $l$:

\begin{itemize}
	\item[]
	\begin{itemize}
		\item[$l=0$] There is an update of source $i$ that arrives to queue $j$. For this case, the incoming
		update replaces the update in the last position of the queue and the age of the incoming update 
		is set to zero, i.e., $x_{jM}^\prime=0$.
		\item[$l=1$] There is an update of another source that arrives to queue $j$. For this case, the incoming
		update replaces the update in the last position of the queue and the age of the incoming update 
		is set to the same value as the age of penultimate update, , i.e., $x_{jM}^\prime=x_{jN}$.
		\\
		\begin{remark}
			If $N=0$, the last update in the queue is an update that is in service. Therefore, 
			when an update of other sources arrives to the system, the incoming
			update replaces the update in service and $x_{j1}^\prime=x_0$.
		\end{remark}
		\item[$l=2$] The update in service in queue $j$ is delivered and therefore the age of the monitor
		changes to $x_{j1}$. Besides, all the elements in the queue move a position ahead, which causes
		that their ages change respectively from $x_{j1}$ to $x_{j2}$, from $x_{j2}$ to $x_{j3}$, $\dots$ and from $x_{jN}$ to $x_{jM}$. Finally, in the last position of the queue, we put a fake update whose age value
		is set to $x_{jM}$, that is, the age of the penultimate element in the queue, i.e., $x_{jM}^\prime=x_{jM}$. 
	\end{itemize}
\end{itemize}

\textcolor{black}{As we have just mentioned, when an update of queue $i$ is delivered to the monitor, 
a fake update is put in the last position of the queue. We now explain that these additional fake updates 
lead to a larger sojourn time of the incoming updates, which implies that 
the overall average AoI is larger and, as a result, the result we obtain provides an upper bound 
of the real average AoI of the system (i.e. without fake updates). Consider that all the updates of queue $i$ 
are delivered before the arrival of a new update to that queue. According to the above explained system, 
the queue is full of fake updates and, therefore, when a new update arrives to the system, it is enqueued in the last position
of the queue. However, in a system without fake updates, a new update would find the queue empty and would start the service upon arrival. As a result, we have that the sojourn time of new updates in a system with fake updates (i.e., the above presented system) 
is clearly larger than the sojourn time without fake updates.} 

Since the Markov Chain is formed by a single state, the stationary distribution is trivial. We define
the vector $\mathbf v=[v_0 \ v_{11}\ v_{12}\ v_{1M}\ v_{21}\  \dots \ v_{K1} \dots v_{KM}]$ and 
also $\mathbf b$ as the vector of size $1+K*M$ with all ones. From the result of Theorem 4 in 
\cite{YK19} and the above reasoning, we know that an upper bound of the AoI is 
given by $v_0$, that is, the first coordinate of the vector $\mathbf v$.

In the remainder of the proof, we present 
the system of equations that $\mathbf v$ satisfies and solve it. 
We first present the equation of the first coordinate of $\mathbf v$:
\begin{multline*}
\sum_{j=1}^K(\lambda_{1}p_{1j}+\sum_{k>1}\lambda_{k}p_{kj}+\mu_j)v_0=\\1+\sum_{j=1}^K\left((\lambda_{1}p_{1j}+\sum_{k>1}\lambda_{k}p_{kj})v_0+\mu_jv_{j1}\right),
\end{multline*}
which can be alternatively written as
\begin{equation}
v_0\sum_{j=1}^K\mu_j=1+\sum_{j=1}^K\mu_jv_{j1}.
\label{eq:v0}
\end{equation}

Let $l^\prime=l+1. $We now present that, for all $j=1,\dots,K$ and all $l=1,\dots,M$, the following equation is satisfied:
\begin{align}
&\sum_{m=1}^K(\lambda_{1}p_{1m}+\sum_{k>1}\lambda_{k}p_{km}+\mu_m)v_{ml}=1+\nonumber\\
&\sum_{m\neq j}(\lambda_{1}p_{1m}+ \sum_{k>1}\lambda_{k}p_{km}+\mu_m)v_{ml}+\nonumber\\
&(\lambda_{1}p_{1j}+\sum_{k>1}\lambda_{k}p_{kj})v_{jl}+\mu_jv_{ml^\prime}\nonumber\\
&\iff \mu_{j}v_{jl}=1+\mu_jv_{jl^\prime}\label{eq:relation-expr-proof-ub-2}
\end{align}

Using recursively the last expression for $l$ equals $1$ to $N$, we get that 
\begin{equation}
\mu_{j}v_{j1}=N-1+\mu_jv_{mN}.
\label{eq:recursion-proof-ub}
\end{equation}

We now focus on the last position of queue $j$ and the equation that it  must satisfy is the following:
\begin{align*}
&\sum_{m=1}^K(\lambda_{1}p_{1m}+\sum_{k>1}\lambda_{k}p_{km}+\mu_m)v_{mM}=1+\\
&\sum_{m\neq j}(\lambda_{1}p_{1m}+\sum_{k>1}\lambda_{k}p_{km}+\mu_m)v_{mM}+
\sum_{k>1}\lambda_{k}p_{kj}v_{jN}+\mu_jv_{jM}\\
&\iff (\lambda_{1}p_{1j}+\sum_{k>1}\lambda_{k}p_{kj})v_{jM}=1+\sum_{k>1}\lambda_{k}p_{kj}v_{jN}
\end{align*}

Besides, from \eqref{eq:relation-expr-proof-ub-2}, for l=N, we have that
$$
\mu_{j}v_{jN}=1+\mu_jv_{jM}.
$$

Using the last two expressions, we get that
\begin{align*}
\mu_{j}v_{jN}&=1+\mu_jv_{jM}\\
&=1+\mu_j\frac{1+\sum_{k>1}\lambda_{k}p_{kj}v_{jN}}{\lambda_{1}p_{1j}+\sum_{k>1}\lambda_{k}p_{kj}}.
\end{align*}

The last expression is equivalent to the following one:
$$
\mu_jv_{jN}\left(1-\frac{\lambda_{2}p_{2j}}{\lambda_{1}p_{1j}+\sum_{k>1}\lambda_{k}p_{kj}}\right)=1+\frac{\mu_j}{\lambda_{1}p_{1j}
	+\sum_{k>1}\lambda_{k}p_{kj}} \iff
$$
$$
\mu_jv_{jN}\left(\frac{\lambda_{1p_{1j}}}{\lambda_{1p_{1j}}+\sum_{k>1}\lambda_{k}p_{kj}}\right)=
\frac{\lambda_{1}p_{1j}+\sum_{k>1}\lambda_{k}p_{kj}+\mu_j}{\lambda_{1}p_{1j}+\sum_{k>1}\lambda_{k}p_{kj}}\iff 
$$
$$
\mu_jv_{jN}\lambda_{1}p_{1j}=\lambda_{1}p_{1j}+\sum_{k>1}\lambda_{k}p_{kj}+\mu_j\iff 
$$
$$
\mu_jv_{jN}=1+\frac{\sum_{k>1}\lambda_{k}p_{kj}+\mu_j}{\lambda_{1}p_{1j}}. 
$$

Using the last expression with \eqref{eq:recursion-proof-ub} and \eqref{eq:v0}, the desired result follows for $i=1$.

\begin{figure}[t!]
	\centering
	\includegraphics[width=\columnwidth,clip=true,trim=40pt 210pt 20pt 210pt]{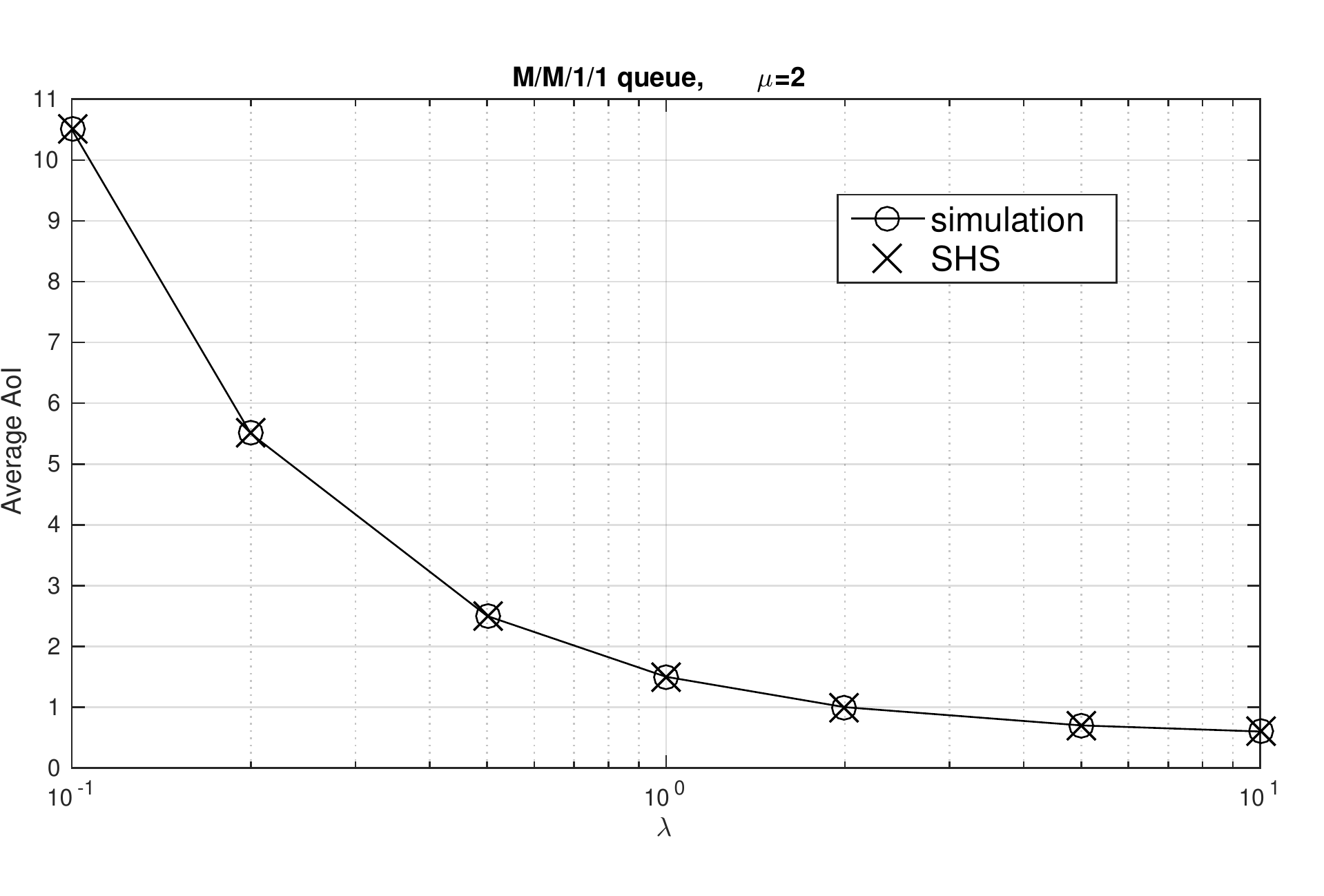}
	\caption{\textcolor{black}{ Simulation results of the M/M/1/1 queue with different values of arrival rates; x-axis in logarithmic scale.}}
	\label{fig:sim-mm11}
\end{figure}
\begin{figure}[t!]
	\centering
	\includegraphics[width=\columnwidth,clip=true,trim=40pt 210pt 20pt 210pt]{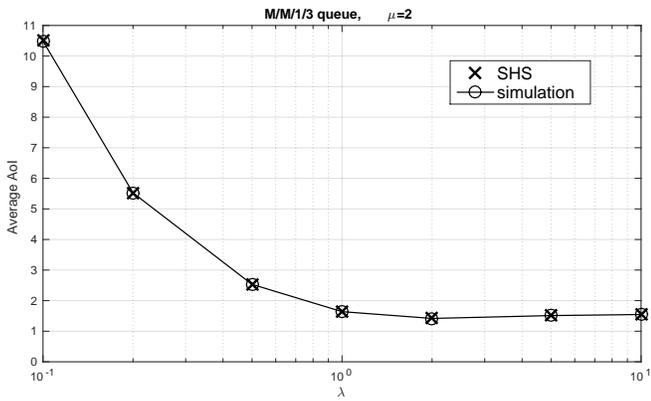}
	\caption{\textcolor{black}{Simulation results of the M/M/1/3* queue with different values of arrival rates; x-axis in logarithmic scale.}}
	\label{fig:sim-mm13}
\end{figure}

\textcolor{black}{
\section{Simulations for M/M/1/1 and  M/M/1/3* queues}
We provide in this section a comparison between the results obtained by SHS and those obtained by simulation to show the accuracy of the SHS method. We consider two different examples with either M/M/1/1 or  M/M/1/3* queues. We consider a system with a single source and without packet losses. We set $\mu=2$. As expected, the curves presented in Figures~\ref{fig:sim-mm11} and ~\ref{fig:sim-mm13} show that, in all the considered cases, the results obtained by simulation coincide with those of SHS. 
}

\section{Table of transitions of Figure~\ref{fig:mm12-routing}}
\label{app:table-mm12-routing}
\begin{table*}[t!]
	\centering
	\small 
	\begin{tabular}{l l l l l}
		$l$ & $q_l\to q_{l^\prime}$ & $\lambda^{l}$ & $x^\prime=x A_l$ & $\bar v_{q_l}A_l$\\\hline
		0 & $00\to 01$ & $\lambda_1p_{12}$ & $[x_0\ 0\ 0\ 0\ 0]$ & $[v_{00}(0)\ 0\ 0\ 0\ 0]$\\
		&  & $\sum_{k>1}\lambda_kp_{k2}$ & $[x_0\ 0\ 0\ x_0\ 0]$ & $[v_{00}(0)\ 0\ 0\ v_{00}(0)\ 0]$\\\hline
		1 & $01\to 00$ & $\mu_{2}$ & $[x_{21}\ 0\ 0\ 0\ 0]$ & $[v_{01}(3)\ 0\ 0\ 0\ 0]$\\
		&  & $\theta_{2}$ & $[x_0\ 0\ 0\ 0\ 0]$ & $[v_{01}(0)\ 0\ 0\ 0\ 0]$\\\hline
		2 & $01\to 02$ & $\lambda_1p_{12}$ & $[x_0\ 0\ 0\ x_{21}\ 0]$ & $[v_{01}(0)\ 0\ 0\ 0\ 0]$\\
		&  & $\sum_{k>1}\lambda_kp_{k2}$ & $[x_0\ 0\ 0\ x_{21}\ x_{21}]$ & $[v_{01}(0)\ 0\ 0\ v_{00}(0)\ 0]$\\\hline
		3 & $02\to 01$ & $\mu_{2}$ & $[x_{21}\ 0\ 0\ x_{22}\ 0]$ & $[v_{02}(3)\ 0\ 0\ v_{02}(4)\ 0]$\\
		&  & $\theta_{2}$ & $[x_0\ 0\ 0\ x_{22}\ 0]$ & $[v_{02}(0)\ 0\ 0\ v_{02}(4)\ 0]$\\\hline
		4 & $02\to 02$ & $\lambda_1p_{12}$ & $[x_0\ 0\ 0\ x_{21}\ 0]$ & $[v_{02}(0)\ 0\ 0\ v_{02}(3)\ 0]$\\
		& & $\sum_{k>1}\lambda_kp_{k2}$ & $[x_0\ 0\ 0\ x_{21}\ x_{21}]$ & $[v_{02}(0)\ 0\ 0\ v_{02}(3)\ v_{02}(3)]$\\\hline
		5 & $00\to 10$ & $\lambda_1p_{11}$ & $[x_0\ 0\ 0\ 0\ 0]$ & $[v_{00}(0)\ 0\ 0\ 0\ 0]$\\
		&  & $\sum_{k>1}\lambda_kp_{k1}$ & $[x_0\ x_0\ 0\ 0\ 0]$ & $[v_{00}(0)\ v_{00}(0)\ 0\ 0\ 0\ 0]$\\\hline
		6 & $10\to 00$ & $\mu_{1}$ & $[x_{12}\ 0\ 0\ 0\ 0]$ & $[v_{01}(1)\ 0\ 0\ 0\ 0]$\\
		&  & $\theta_{1}$ & $[x_0\ 0\ 0\ 0\ 0]$ & $[v_{01}(0)\ 0\ 0\ 0\ 0]$\\\hline
		7 & $01\to 11$ & $\lambda_1p_{11}$ & $[x_0\ 0\ 0\ x_{21}\ 0]$ & $[v_{01}(0)\ 0\ 0\ v_{01}(3)\ 0]$\\
		&  & $\sum_{k>1}\lambda_kp_{k1}$ & $[x_0\ x_0\ 0\ x_{21}\ 0]$ & $[v_{01}(0)\ v_{01}(0)\ 0\ 0\ v_{01}(3)\ 0]$\\\hline
		8 & $11\to 01$ & $\mu_{1}$ & $[x_{11}\ 0\ 0\ x_{21}\ 0]$ & $[v_{11}(1)\ 0\ 0\ v_{11}(3)\ 0]$\\
		&  & $\theta_{1}$ & $[x_0\ 0\ 0\ x_{21}\ 0]$ & $[v_{11}(0)\ 0\ 0\ v_{11}(3)\ 0]$\\\hline
		9& $02\to 12$ & $\lambda_1p_{11}$ & $[x_0\ 0\ 0\ x_{21}\ x_{22}]$ & $[v_{02}(0)\ 0\ 0\ v_{02}(3)\ v_{02}(4)]$\\
		&  & $\sum_{k>1}\lambda_kp_{k1}$ & $[x_0\ x_0\ 0\ x_{21}\ x_{22}]$ & $[v_{02}(0)\ v_{02}(0)\ 0\ 0\ v_{02}(3)\ v_{02}(4)]$\\\hline
		10 & $12\to 02$ & $\mu_{1}$ & $[x_{11}\ 0\ 0\ x_{21}\ x_{22}]$ & $[v_{12}(1)\ 0\ 0\ v_{12}(3)\ v_{12}(4)]$\\
		&  & $\theta_{1}$ & $[x_0\ 0\ 0\ x_{21}\ x_{22}]$ & $[v_{12}(0)\ 0\ 0\ v_{12}(3)\ v_{12}(4)]$\\\hline
		11 & $10\to 11$ & $\lambda_1p_{12}$ & $[x_0\ x_{11}\ 0\ 0\ 0]$ & $[v_{10}(0)\ v_{10}(1)\ 0\ 0\ 0]$\\
		&  & $\sum_{k>1}\lambda_kp_{k2}$ & $[x_0\ x_{11}\ 0\ x_0\ 0]$ & $[v_{10}(0)\ v_{10}(1)\ 0\ v_{10}(0)\ 0]$\\\hline
		12 & $11\to 10$ & $\mu_{2}$ & $[x_{21}\ x_{11}\ 0\ 0\ 0]$ & $[v_{11}(3)\ v_{11}(1)\ 0\ 0\ 0]$\\
		&  & $\theta_{2}$ & $[x_0\ x_{11}\ 0\ 0\ 0]$ & $[v_{11}(0)\ v_{11}(1)\ 0\ 0\ 0]$\\\hline
		13 & $11\to 12$ & $\lambda_1p_{12}$ & $[x_0\ x_{11}\ 0\ x_{21}\ 0]$ & $[v_{11}(0)\ v_{11}(1)\ 0\ v_{11}(3)\ 0]$\\
		&  & $\sum_{k>1}\lambda_kp_{k2}$ & $[x_0\ x_{11}\ 0\ x_{21}\ x_{21}]$ & $[v_{11}(0)\ v_{11}(1)\ 0\ v_{11}(3)\ v_{11}(3)]$\\\hline
		14 & $12\to 11$ & $\mu_{2}$ & $[x_{21}\ 0\ 0\ x_{22}\ 0]$ & $[v_{12}(3)\ 0\ 0\ v_{12}(4)\ 0]$\\
		&  & $\theta_{2}$ & $[x_0\ 0\ 0\ x_{22}\ 0]$ & $[v_{12}(0)\ 0\ 0\ v_{12}(4)\ 0]$\\\hline
		15 & $12\to 12$ & $\lambda_1p_{12}$ & $[x_0\ x_{11}\ 0\ x_{21}\ 0]$ & $[v_{12}(0)\ v_{12}(1)\ 0\ v_{12}(3)\ 0]$\\
		& & $\sum_{k>1}\lambda_kp_{k2}$ & $[x_0\ x_{11}\ 0\ x_{21}\ x_{21}]$ & $[v_{12}(0)\ v_{12}(1)\ 0\ v_{12}(3)\ v_{12}(3)]$\\\hline
		16 & $10\to 20$ & $\lambda_1p_{11}$ & $[x_0\ x_{11}\ 0\ 0\ 0]$ & $[v_{10}(0)\ x_{10}(1)\ 0\ 0\ 0]$\\
		&  & $\sum_{k>1}\lambda_kp_{k1}$ & $[x_0\ x_{11}\ x_{11}\ 0\ 0]$ & $[v_{10}(0)\ v_{10}(1)\ v_{10}(1)\ 0\ 0\ 0]$\\\hline
		17 & $20\to 10$ & $\mu_{1}$ & $[x_{12}\ x_{22}\ 0\ 0\ 0]$ & $[v_{20}(1)\ v_{20}(2)\ 0\ 0\ 0]$\\
		&  & $\theta_{1}$ & $[x_0\ x_{22}\ 0\ 0\ 0]$ & $[v_{20}(0)\ x_{20}(2)\ 0\ 0\ 0]$\\\hline
		18 & $11\to 21$ & $\lambda_1p_{11}$ & $[x_0\ x_{11}\ 0\ x_{21}\ 0]$ & $[v_{11}(0)\ v_{11}(1)\ 0\ v_{11}(3)\ 0]$\\
		&  & $\sum_{k>1}\lambda_kp_{k1}$ & $[x_0\ x_{11}\ x_{11}\ x_{21}\ 0]$ & $[v_{11}(0)\ v_{11}(1)\ v_{11}(1)\ v_{11}(3)\ 0]$\\\hline
		19 & $21\to 11$ & $\mu_{1}$ & $[x_{11}\ x_{12}\ 0\ x_{21}\ 0]$ & $[v_{21}(1)\ v_{21}(2)\ 0\ v_{21}(3)\ 0]$\\
		&  & $\theta_{1}$ & $[x_0\ x_{12}\ 0\ x_{21}\ 0]$ & $[v_{21}(0)\ v_{21}(2)\ 0\ v_{11}(3)\ 0]$\\\hline
		20& $12\to 22$ & $\lambda_1p_{11}$ & $[x_0\ x_{11}\ 0\ x_{21}\ x_{22}]$ & $[v_{12}(0)\ v_{12}(1)\ 0\ v_{12}(3)\ v_{12}(4)]$\\
		&  & $\sum_{k>1}\lambda_kp_{k1}$ & $[x_0\ x_{11}\ 0\ x_{21}\ x_{22}]$ & $[v_{12}(0)\ v_{12}(1)\ v_{12}(1)\ v_{12}(3)\ v_{12}(4)]$\\\hline
		21 & $22\to 12$ & $\mu_{1}$ & $[x_{11}\ x_{12}\ 0\ x_{21}\ x_{22}]$ & $[v_{22}(1)\ v_{22}(2)\ 0\ v_{22}(3)\ v_{12}(4)]$\\
		&  & $\theta_{1}$ & $[x_0\ 0\ 0\ x_{21}\ x_{22}]$ & $[v_{22}(0)\ v_{22}(2)\ 0\ v_{22}(3)\ v_{22}(4)]$\\\hline
		22 & $20\to 20$ & $\lambda_1p_{11}$ & $[x_{0}\ x_{11}\ 0\ 0 \ 0]$ & $[v_{20}(0)\ v_{20}(1)\ 0\ 0\ 0]$\\
		&  & $\sum_{k>1}\lambda_kp_{k1}$ & $[x_0\ x_{11}\ x_{11}\ x_{21}\ 0]$ & $[v_{22}(0)\ v_{20}(1)\ v_{20}(1)\ 0 \  0\ 0]$\\\hline
		23 & $20\to 21$ & $\lambda_1p_{12}$ & $[x_0\ x_{11}\ x_{12}\ 0\ 0]$ & $[v_{20}(0)\ v_{20}(1)\ v_{20}(2)\ 0\ 0]$\\
		&  & $\sum_{k>1}\lambda_kp_{k2}$ & $[x_0\ x_{11}\ x_{12}\ x_0\ 0]$ & $[v_{20}(0)\ v_{20}(1)\ v_{20}(2)\ v_{20}(0)\ 0]$\\\hline
		24 & $21\to 20$ & $\mu_{2}$ & $[x_{21}\ x_{11}\ x_{12}\ 0\ 0]$ & $[v_{21}(3)\ v_{21}(1)\ v_{21}(2)\ 0\ 0]$\\
		&  & $\theta_{2}$ & $[x_0\ x_{11}\ x_{12}\ 0\ 0]$ & $[v_{21}(0)\ v_{21}(1)\ v_{21}(2)\ 0\ 0]$\\\hline
		25 & $21\to 21$ & $\lambda_1p_{11}$ & $[x_{0}\ x_{11}\ 0\ x_{21} \ 0]$ & $[v_{21}(0)\ v_{21}(1)\ 0\ v_{21}(3)\ 0]$\\
		&  & $\sum_{k>1}\lambda_kp_{k1}$ & $[x_0\ x_{11}\ x_{11}\ x_{21}\ 0]$ & $[v_{21}(0)\ v_{21}(1)\ v_{21}(1)\  v_{21}(3)\ 0]$\\\hline
		26 & $21\to 22$ & $\lambda_1p_{12}$ & $[x_0\ x_{11}\ x_{12}\ x_{21}\ 0]$ & $[v_{21}(0)\ v_{21}(1)\ v_{21}(2)\ x_{21}(3)\ 0]$\\
		&  & $\sum_{k>1}\lambda_kp_{k2}$ & $[x_0\ x_{11}\ x_{12}\ x_{21}\ x_{21}]$ & $[v_{21}(0)\ v_{21}(1)\ v_{21}(2)\ v_{21}(3)\ v_{21}(3)]$\\\hline
		27 & $22\to 21$ & $\mu_{2}$ & $[x_{21}\ x_{11}\ x_{12}\ x_{21}\ 0]$ & $[v_{22}(3)\ v_{22}(1)\ v_{22}(2)\ x_{22}(4)\ 0]$\\
		&  & $\theta_{2}$ & $[x_0\ x_{11}\ x_{12}\ x_{21}\ 0]$ & $[v_{22}(0)\ v_{22}(1)\ v_{22}(2)\ v_{22}(4)\ 0]$\\\hline
		28 & $22\to 22$ & $\lambda_1p_{11}$ & $[x_{0}\ x_{11}\ 0\ x_{21} \ x_{22}]$ & $[v_{22}(0)\ v_{22}(1)\ 0\ v_{22}(3)\ v_{22}(4)]$\\
		&  & $\sum_{k>1}\lambda_kp_{k1}$ & $[x_0\ x_{11}\ x_{11}\ x_{21}\ x_{22}]$ & $[v_{22}(0)\ v_{22}(1)\ v_{22}(1)\  v_{22}(3)\ v_{22}(4)]$\\\hline
		29 & $22\to 22$ & $\lambda_1p_{12}$ & $[x_{0}\ x_{11}\ x_{12}\ x_{21} \ 0]$ & $[v_{22}(0)\ v_{22}(1)\ v_{22}(2)\ v_{22}(3)\ 0]$\\
		&  & $\sum_{k>1}\lambda_kp_{k2}$ & $[x_0\ x_{11}\ x_{12}\ x_{21}\ x_{21}]$ & $[v_{22}(0)\ v_{22}(1)\ v_{22}(2)\  v_{22}(3)\ v_{22}(3)]$\\\hline 
	\end{tabular}
	\caption{Table of transitions of Figure~\ref{fig:mm12-routing}.}
	\label{tab:mm12-routing}
\end{table*}

\end{document}